\documentclass[journal]{IEEEtran}
\usepackage[utf8]{inputenc}
\usepackage{biblatex}
\usepackage{amsmath}
\usepackage{amssymb}
\usepackage{systeme}
\usepackage{graphicx}
\usepackage{caption}
\usepackage{cleveref}
\usepackage{bm}
\usepackage[linesnumbered,ruled,vlined]{algorithm2e}
\usepackage{amsthm}
\usepackage{spalign}
\usepackage{amsfonts}
\usepackage{enumerate}

\newtheorem{theorem}{Theorem}
\newtheorem{lemma}{Lemma}

\newtheorem{proposal}{Proposal}

\graphicspath{ {./images/} }

\let\oldnl\nl% Store \nl in \oldnl
\newcommand{\nonl}{\renewcommand{\nl}{\let\nl\oldnl}}% Remove line number for one line

\newcommand{\by}{\mathbf{y}}

\newcommand{\bx}{\mathbf{x}}
\newcommand{\bw}{\mathbf{w}}
\newcommand{\bh}{\mathbf{h}}
\newcommand{\bq}{\mathbf{q}}

\newcommand{\br}{\mathbf{r}}

\newcommand{\bz}{\mathbf{z}}
\newcommand{\bb}{\mathbf{b}}
\newcommand{\bv}{\mathbf{v}}
\newcommand{\bd}{\mathbf{d}}
\newcommand{\bff}{\mathbf{f}}
\newcommand{\bg}{\mathbf{g}}
\newcommand{\bn}{\mathbf{n}}
\newcommand{\bzero}{\mathbf{0}}
\newcommand{\bmu}{\bm{\mu}}

\newcommand{\bA}{\mathbf{A}}
\newcommand{\bI}{\mathbf{I}}
\newcommand{\bU}{\mathbf{U}}
\newcommand{\bV}{\mathbf{V}}
\newcommand{\bS}{\mathbf{S}}
\newcommand{\bW}{\mathbf{W}}
\newcommand{\bH}{\mathbf{H}}
\newcommand{\bs}{\mathbf{s}}
\newcommand{\bF}{\mathbf{F}}

\newcommand{\bP}{\mathbf{P}}
\newcommand{\bJ}{\mathbf{J}}
\newcommand{\bD}{\mathbf{D}}
\newcommand{\bM}{\mathbf{M}}
\newcommand{\bmm}{\mathbf{m}}
\newcommand{\bB}{\mathbf{B}}
\newcommand{\bQ}{\mathbf{Q}}
\newcommand{\bZ}{\mathbf{Z}}

\newcommand{\halfa}{\Hat{\alpha}}
\newcommand{\talfa}{\Tilde{\alpha}}

\newcommand{\bhs}{\Hat{\mathbf{s}}}
\newcommand{\bhf}{\Hat{\mathbf{f}}}
\newcommand{\bhq}{\Hat{\mathbf{q}}}
\newcommand{\bhb}{\Hat{\mathbf{b}}}

\newcommand{\bhphi}{\Hat{\bm{\phi}}}
\newcommand{\bhnu}{\Hat{\bm{\nu}}}
\newcommand{\bhu}{\Hat{\mathbf{u}}}

\newcommand{\bbq}{\overline{\mathbf{q}}}
\newcommand{\bbf}{\overline{\mathbf{f}}}
\newcommand{\bbs}{\overline{\mathbf{s}}}

\newcommand{\bbr}{\overline{\mathbf{\br}}}
\newcommand{\bbh}{\overline{\mathbf{\bh}}}

\newcommand{\inv}[1]{\frac{1}{#1}}

\newcommand{\normDensity}{\mathcal{N}}
\newcommand{\expectation}{\mathbb{E}}

\SetKwInput{KwInput}{Input}               
\SetKwInput{KwOutput}{Output}
\SetKwInOut{KwInitialization}{Initialization}
\SetKwBlock{KwBlockA}{Block A}{end}
\SetKwBlock{KwBlockB}{Block B}{end}

\title{Divergence Estimation in Message Passing algorithms}
\author{Nikolajs Skuratovs, 
Michael Davies, ~\IEEEmembership{Fellow IEEE}\thanks{This work was supported by the ERC project C-SENSE (ERC-ADG-2015-694888). MD is also supported by a Royal Society Wolfson Research Merit Award.}}

\begin{document}

\maketitle

\begin{abstract}
    Many modern imaging applications can be modeled as compressed sensing linear inverse problems. When the measurement operator involved in the inverse problem is sufficiently random, denoising Scalable Message Passing (SMP) algorithms have a potential to demonstrate high efficiency in recovering compressed data. One of the key components enabling SMP to achieve fast convergence, stability and predictable dynamics is the Onsager correction that must be updated at each iteration of the algorithm. This correction involves the denoiser's divergence  that is traditionally estimated via the Black-Box Monte Carlo (BB-MC) method \cite{MC-divergence}. While the BB-MC method demonstrates satisfying accuracy of estimation, it requires executing the denoiser additional times at each iteration and might lead to a substantial increase in computational cost of the SMP algorithms. In this work we develop two Large System Limit models of the Onsager correction for denoisers operating within SMP algorithms and use these models to propose two practical classes of divergence estimators that require no additional executions of the denoiser and demonstrate similar or superior correction compared to the BB-MC method. 
\end{abstract}


\section{Introduction}

In this work we consider a particular sub-problem that arises in certain iterative methods designed to recover a signal $\bx \in \mathbb{R}^N$ from a set of linear measurements 

\noindent
\begin{equation}
    \by = \bA \bx + \bw
    \label{eq:y_measurements}
\end{equation}

\noindent
where $\by \in \mathbb{R}^M$ is the set of measurements, $\bw \in \mathbb{R}^M$  is a zero-mean i.i.d. Gaussian noise vector $\bw \sim \normDensity(0,v_w \bI_M)$ and $\bA \in \mathbb{R}^{M \times N}$ is a measurement matrix that is assumed to be available. We consider the large scale compressed sensing scenario $M < N$ with a subsampling factor $\delta = \frac{M}{N} = O(1)$. 

While there are many first-order iterative methods for recovering $\bx$ from the set of measurement \eqref{eq:y_measurements} including \cite{Matching_Pursuits}, \cite{Orthogonal_Matching_Pursuits}, \cite{Hard-thresholding}, \cite{ISTA} and many others, in this work we focus on the family of \textit{Scalable Message Passing (SMP)} algorithms that includes  Approximate Message Passing (AMP) \cite{AMP}, Orthogonal AMP (OAMP) \cite{OAMP},  Vector AMP (VAMP) \cite{VAMP}, Conjugate Gradient VAMP (CG-VAMP) \cite{OurSecondPaper}, \cite{CG_EP}, \cite{OurPaper}, Warm-Started CG-VAMP (WS-CG-VAMP) \cite{OurSecondPaper}, Convolutional AMP (CAMP) \cite{CAMP} and others. When the measurement operator $\bA$ comes from a certain family of random matrices, which may be different for each example of SMP, these algorithms demonstrate high per-iteration improvement and stable and predictable dynamics. Additionally, it is evidenced that SMP algorithms can recover complex signals like natural images by employing powerful Plug-and-Play (PnP) denoisers like BM3D \cite{BM3D}, Non-Local Means \cite{NLM}, Denoising CNN \cite{NS-VAMP} and others, and demonstrate State-of-The-Art performance for certain examples of $\bA$ \cite{D-AMP}. 

On a general level, an SMP algorithm is an iterative method with a linear step followed by a denoising step. It can be shown \cite{UnifiedSE}, \cite{NS-VAMP}, \cite{SE_AMP}, \cite{CAMP} that one can be flexible with the choice of denoisers in SMP as long as the key ingredient, the divergence of the denoiser at each iteration, can be computed to form a so-called \textit{Onsager Correction} for the denoiser. In the literature on SMP algorithms \cite{D-AMP}, \cite{OurSecondPaper}, \cite{MRI_VAMP_2020}, \cite{D-OAMP}, \cite{D-VAMP}, \cite{NS-VAMP} and others, the suggested method for computing the divergence of a PnP denoiser is the Black-Box Monte Carlo (BB-MC) method \cite{MC-divergence}. The BB-MC method computes an estimate of the divergence of a function $\bff(\bx)$ that admits a well-defined second-order Taylor expansion by executing this function again at point $\bx + \epsilon \bn$ with the scalar $\epsilon$ approaching zero and where $\bn$ is a zero-mean i.i.d. random vector with unit variance and finite higher order moments. Then one can show that the divergence $\inv{N} \nabla_{\bx} \cdot \bff(\bx) = \inv{N} \sum_{i = 1}^N \frac{\bff(\bx)}{\partial x_i}$ of $\bff$ is equivalent to \cite{MC-divergence}

\noindent
\begin{equation}
    \inv{N} \nabla_{\bx} \cdot \bff(\bx) = \lim_{\epsilon \rightarrow 0} \expectation_{\epsilon} \Bigg[  \bn^T \bigg( \frac{\bff(\bx + \epsilon \bn) - \bff(\bx)}{\epsilon} \bigg) \Bigg]
    \label{eq:black_box_div_estimator}
\end{equation}

\noindent
To approximate the expectation operator in \eqref{eq:black_box_div_estimator}, one can use MC trials and implement the inner product inside of the expectation multiple times and average the results. However, given that the function $\bff$ is of the appropriate class and the dimension of $\bx$ is sufficiently large, one can obtain a satisfactory accuracy of divergence estimation with only a single trial.

While this approach provides a practical method for the divergence estimation and leads to stable dynamics of SMP algorithms, it has two drawbacks. First, it assumes that the chosen denoiser $\bff$ admits a well-defined second-order Taylor expansion, which is not the case for denoisers like BM3D and for ReLU based CNNs \cite{NS-VAMP} that involves non-linear operations like thresholding as subroutines. This violation might result in unsatisfactory accuracy of the estimation and lead to the necessity for additional MC trials. Additionally one can no longer use too small values of $\epsilon$ as in this case the estimator \eqref{eq:black_box_div_estimator} becomes unstable \cite{MC-divergence}, which leads to the necessity to tune this parameter very carefully and, to the best of our knowledge, there is no rigorous method for this. 

The second problem with the BB-MC method is that it requires executing the denoiser once again or even multiple times, if one needs higher accuracy of the estimate. When the dimension of the inverse problem is large, as in modern computational imaging problems, executing powerful denoisers can be the dominant cost of the algorithm and it is desired to execute it as infrequently as possible.

In this work, we leverage the properties of the SMP algorithms to develop rigorous Large System Limit models for the divergence of a denoiser. We use the developed models to propose two divergence estimation techniques applicable for any SMP algorithm following the State Evolution (SE), although we also demonstrate numerically that the methods are stable and accurate even for algorithms violating such dynamics. The first method allows estimating the divergence of the denoiser at a cost dominated by one inner-product with the output of the denoiser and works as a black-box tool. Such a fast estimator can be used to optimize the denoiser using the SURE technique \cite{SURE} to achieve the optimal performance of SMP algorithms. Yet, this method is less robust with respect to the decreased dimensionality of the inverse problem $N$ and $M$. Thus, we propose the second method that demonstrates stable and accurate performance even for dimensions $N$ of order $10^4$ and leads to superior performance of SMP algorithms compared to the case where a BB-MC estimator is used. The cost of the second method is dominated by one matrix-vector product with the measurement operator $\bA$. We numerically compare the proposed methods against the BB-MC method in the context of AMP, MF-OAMP, CG-VAMP and WS-CG-VAMP used for recovering natural images from compressed measurements.

\subsection{Notations}

We use roman $v$ for scalars, small boldface $\bv$ for vectors and capital boldface $\bV$ for matrices. We frequently use the identity matrix $\bI_N$ with a subscript to define that this identity matrix is of dimension $N$ or without a subscript where the dimensionality is clear from the context. We define $Tr\big\{ \bM \big\}$ to be the trace of a matrix $\bM$ and $\kappa(\bM)$ to be the condition number of $\bM$. We use $||\cdot||_k$ to define the $l_k$ norm and $||\cdot||$ specifically for $l_2$ norm. The divergence of a function $f(\bx)$ with respect to the vector $\bx$ is defined as $\nabla_{\bx} \cdot f(\bx) = \sum_{i = 1}^N \frac{f(\bx)}{\partial x_i}$. By writing $q(x) = \normDensity(\bmm, \Sigma)$ we mean that the density $q(x)$ is normal with mean vector $\bmm$ and covariance matrix $\Sigma$. We reserve the letter $t$ for the outer-loop iteration number of the EP- and VAMP-based algorithms and use the letter $i$ for the iteration number of the Conjugate Gradient algorithm. We mark an estimate of some stochastic quantities such as $v$ with a tilde, i.e. $\Tilde{v}$. Lastly, we use notation \textit{i.i.d.} for a shorthand of \textit{independent and identically distributed}.

\section{Background on SMP algorithms}

In this section we briefly review the structure and the main properties of SMP algorithms to set up the context of the paper. For more details, please refer to \cite{OurSecondPaper}, \cite{NS-VAMP}, \cite{UnifiedSE}, \cite{D-AMP}, \cite{OAMP} and the references therein. On a general level, any SMP algorithm can be written as an iterative process involving the following two steps

\noindent
\begin{align}
    LB: \quad & \br_t = \bff_L \big( \by, \bs_t, \bs_{t-1}, ..., \bs_0 \big) \label{eq:LB} \\
    DB: \quad & \bs_{t+1} = \bff_D \big( \br_t \big) \label{eq:DB}
\end{align}

\noindent
with the respective oracle error vectors

\noindent
\begin{gather}
    \bh_t = \br_t - \bx \label{eq:h_t} \\
    \bq_t = \bs_t - \bx \label{eq:q_t}
\end{gather}

\noindent
In \eqref{eq:LB} and \eqref{eq:DB}, the notations $LB$ and $DB$ stand for the ``linear block'' and the ``denoising block'', and, depending on the choice of SMP algorithms, one has to construct the function $\bff_L$ and $\bff_D$ in different ways. Next we split the general structure \eqref{eq:LB}-\eqref{eq:DB} into two groups and discuss how the functions $\bff_L$ and $\bff_D$ are designed in each group.

\subsection{OAMP-based algorithms} \label{sec:VAMP-based algorithms}

The first group of algorithms, which we refer as OAMP-based algorithms, constructs $\bff_D$ and $\bff_L$ such that their output errors are asymptotically orthogonal to their input errors \cite{OAMP}. This group of algorithms includes Match Filter OAMP (MF-OAMP) \cite{OAMP}, \cite{D-OAMP}, VAMP \cite{NS-VAMP}, \cite{EP_Keigo}, CG-VAMP \cite{CG_EP}, \cite{OurSecondPaper}, \cite{OurPaper}, WS-CG-VAMP \cite{OurSecondPaper} and others. Here the structure of the function $\bff_D$ from \eqref{eq:DB} is

\noindent
\begin{equation}
    \bff_D(\br_t) = C_t \Big( \bg_D(\br_t) - \alpha_t \br_t \Big) \label{eq:f_d_VAMP}
\end{equation}

\noindent
where $\alpha_t$ is the divergence of a denoiser $\bg_D(\br_t)$ estimating $\bx$ from the intrinsic measurement $\br_t = \bx + \bh_t$ with $\bh_t$ modeled as zero-mean i.i.d. Gaussian with variance $v_{h_t}$ and independent of $\bx$. While in VAMP \cite{VAMP}, CG-VAMP and WS-CG-VAMP \cite{OurSecondPaper} it is suggested to use the scalar $C_t = (1 - \alpha_t)^{-1}$, in the work on Denoising OAMP \cite{D-OAMP} it is proposed to choose $C_t$ that minimizes the Mean Squared Error (MSE) of the resulting denoiser $\bff_D(\br_t)$. Lastly, the term $\alpha_t \br_t$ is the Onsager correction of the denoiser. This term ensures asymptotic orthogonality of the errors $\bq_t$ and $\bh_t$, leads to stable and efficient operation of the algorithm and makes the analysis of the algorithm analytically tractable \cite{UnifiedSE}, \cite{NS-VAMP}, \cite{AMP_convergence_general_A}.

In the MF-OAMP, VAMP and CG-VAMP algorithms, the linear step is dependent only on the last state $\bs_t$ and has the form

\noindent
\begin{equation}
    \bff_L \big( \by, \bs_t \big) = \bs_t + \gamma_t^{-1} \bA^T \bF_t \big( \by - \bA \bs_t \big) \label{eq:f_l_VAMP}
\end{equation}

\noindent
for different choices of $\bF_t$. In MF-OAMP, the matrix $\bF_t$ corresponds to the identity $\bI_M$ \cite{OAMP}, while in VAMP\footnote{In the original work presenting VAMP, it is used a different form of the linear update that can be related to \eqref{eq:f_l_VAMP} through the Woodbury transformation. For more details, please refer to \textit{Appendix B} in \cite{EP_Keigo}.} \cite{NS-VAMP}, \cite{EP_Keigo} the vector $\bF_t \big( \by - \bA \bs_t \big)$ is the solution to the system of linear equations 

\noindent
\begin{equation}
    \bW_t \bmu_t = \bz_t \label{eq:SLE}
\end{equation}

\noindent
for $\mu_t$, where

\noindent
\begin{equation}
    \bz_t = \by - \bA \bs_t
\end{equation}

\noindent
and 

\noindent
\begin{equation}
    \bW_t = v_w \bI_M + v_{q_t} \bA \bA^T \label{eq:W_t}
\end{equation}

\noindent
with $v_{q_t}$ modeling the variance of the error vector $\bq_t$. In CG-VAMP \cite{CG_EP}, \cite{OurSecondPaper} the exact solutions $\bmu_t$ is approximated with $i$ iterations of the zero-initialized Conjugate Gradient (CG) algorithm\footnote{In this work we used the CG algorithm as in \cite{CG_EP}.} that produces an output $\bmu_t^i = \bF_t^i \big( \by - \bA \bs_t \big)$. In WS-CG-VAMP \cite{OurSecondPaper} the CG algorithm uses multiple non-zero initializations that lead to the improved reconstruction performance of the algorithm, but dependence of $\bff_L$ on the whole history of $\bs_{\tau}$, $\tau = 0,...,t$ and loss of the 1D SE property.

The form of the scalar $\gamma_t$ in \eqref{eq:f_l_VAMP} depends on the form of the function $\bF_t$. For MF-OAMP, it is equal to $1$ \cite{OAMP}, while in VAMP it corresponds to \cite{EP_Keigo}

\noindent
\begin{equation}
    \gamma_t = \lim_{N \rightarrow \infty } \inv{N} Tr \big\{ \bA^T \bW_t^{-1} \bA  \big\} \label{eq:gamma_VAMP}
\end{equation}

\noindent
For CG-VAMP, where $\bF_t^i$ represents the CG algorithm with $i$ iterations, one can use Theorem 1 from \cite{OurSecondPaper} to estimate $\gamma_t$ iteratively using only the information generated by CG with the overall cost dominated by $i$ inner-products of $M$ dimensional vectors. The scalar $\gamma_t$ plays a similar role to the divergence of the denoiser $\alpha_t$ in the denoising block and ensures the asymptotic orthogonality between the error vectors $\bh_t$ and $\bq_t$. For OAMP-based algorithms with a multi-dimensional SE as WS-CG-VAMP \cite{OurSecondPaper}, there are multiple scalars $\gamma_t^k$ and the linear update \eqref{eq:f_l_VAMP} involves multiple corrections terms. For more details, please refer to \cite{OurSecondPaper} and \cite{UnifiedSE}.

\subsection{AMP-based algorithms} \label{sec:AMP-based algorithms}

For AMP-based algorithms including AMP \cite{AMP} and CAMP \cite{CAMP}, the structure of the function $\bff_D$ from \eqref{eq:DB} is

\noindent
\begin{equation}
    \bff_D(\br_t) = \bg_D(\br_t) \label{eq:f_d_AMP}
\end{equation}

\noindent
Although here the denoiser is not corrected, in AMP it is compensated by designing the function $\bff_L$ to have the following form

\noindent
\begin{equation}
    \bff_L \big( \by, \bs_t, \bs_{t-1}, ..., \bs_0 \big) = \bs_t + \bA^T \bz_t \label{eq:f_l_AMP}
\end{equation}

\noindent
with $\bz_t$ updated as

\noindent
\begin{equation}
    \bz_t = \by - \bA \bs_t - \alpha_{t-1} \delta^{-1} \bz_{t-1} \label{eq:z_t_AMP}
\end{equation}

\noindent
where $\delta = \frac{M}{N}$ is the subsampling factor and $\alpha_{t-1}$ is the divergence of $\bg_D(\br_{t-1})$. In \eqref{eq:z_t_AMP} the term $\alpha_{t-1} \delta^{-1} \bz_{t-1}$ is the Onsager correction that accounts for both the linear and the denoising steps at once. In CAMP, the function \eqref{eq:f_l_AMP} has the same structure, but the update for $\bz_t$ involves a more complex Onsager correction that includes $t-1$ divergences $\alpha_{\tau}, \tau = 0...t-1$ \cite{CAMP}.

\subsection{Error dynamics of Message Passing algorithms} \label{sec:error_dynamics}

Given that the linear block and the denoising block in AMP- and OAMP-based algorithms are designed as discussed above, the dynamics of the error vector $\bh_t$ and $\bq_t$ from \eqref{eq:h_t} and \eqref{eq:q_t} can be rigorously defined under the following assumptions \cite{UnifiedSE}

\textbf{Assumption 1}: The dimensions of the signal model $N$ and $M$ approach infinity with a fixed ratio $\delta = \frac{M}{N} = O(1)$

\noindent
\textbf{Assumption 2}:
\begin{enumerate}[a)]
    \item For AMP: The measurement matrix $\bA$ is orthogonally invariant, such that in the SVD of $\bA = \bU \bS \bV^T$, the matrices $\bU$ and $\bV$ are independent of other random terms and are uniformly distributed on the set of orthogonal matrices, while the matrix $\bS^T \bS$ has the Limiting Eigenvalue Distribution \cite{rand_mat_methods_book} with the first $2t$ moments equal to the first $2t$ moments of  Mar\u{c}henko-Pastur (MP) distribution \cite{rand_mat_methods_book}, where $t$ is the maximum number of iterations of AMP.
    
    \item For OAMP-based and CAMP: The same condition on $\bV$, while $\bU$ is allowed to be any orthogonal matrix and the matrix $\bS^T \bS$ is allowed to have any Limiting Eigenvalue Distribution with compact support. For those cases, we say $\bA$ is Right-Orthogonally invariant.
\end{enumerate}

\textbf{Assumption 3}: The denoiser $\bg_D$ is uniformly Lipschitz so that the sequence of functions $\bg_D: \mathbb{R}^N \mapsto \mathbb{R}^N$ indexed by $N$ are Lipschitz continuous with a Lipschitz constant $L_N < \infty$ as $N \rightarrow \infty$ \cite{NS-VAMP}, \cite{AMP_SE_non_separable}. Additionally, we assume the sequences of the following inner-products are almost surely finite as $N \rightarrow \infty$ \cite{NS-VAMP}

\noindent
\begin{gather*}
    \lim_{N \rightarrow \infty} \inv{N} \bg_D(\bx + \bd_1)^T \bg_D(\bx + \bd_2), \; \; \lim_{N \rightarrow \infty} \inv{N} \bx^T \bg_D(\bx + \bd_1), \\ 
    \lim_{N \rightarrow \infty} \inv{N} \bd_1^T \bg_D(\bx + \bd_2), \quad \lim_{N \rightarrow \infty} \inv{N} \bx^T \bz_1, \quad  \lim_{N \rightarrow \infty} \inv{N} ||\bx||^2 
\end{gather*}

\noindent
where $\bd_1, \bd_2 \in \mathbb{R}^N$ with $( \bd_{1,n}, \bd_{2,n} ) \sim \normDensity(\bzero, C)$ for some positive definite $C \in \mathbb{R}^2$.

\noindent
Additionally, without the loss of generality we let $\bA$ be normalized so that $\inv{N} Tr \big\{ \bA \bA^T \big\} = 1$. Under these assumptions, both $\bq_t$ and $\bh_t$ can be modeled as zero-mean random vectors with variances

\noindent
\begin{gather*}
    \lim_{N \rightarrow \infty} \inv{N} ||\bq_t||^2 \overset{a.s.}{=} v_{q_t} \\
    \lim_{N \rightarrow \infty} \inv{N} ||\bh_t||^2 \overset{a.s.}{=} v_{h_t}
\end{gather*}

\noindent
Moreover, by defining a vector $\bb_t = \bV^T \bq_t$ we have that \cite{UnifiedSE}, \cite{VAMP}

\noindent
\begin{gather}
    \bh_t \sim \normDensity(0, v_{h_t} \bI_N) \label{eq:h_t_gaussian} \\
    \bb_t \sim \normDensity(0, v_{q_t} \bI_N) \label{eq:b_t_gaussian}
\end{gather}

\noindent
where 

\noindent
\begin{gather}
    \lim_{N \rightarrow \infty} \inv{N} \bh_{\tau}^T \bq_{k} \overset{a.s.}{=} 0  \label{eq:h_q_independence} \\
    \lim_{N \rightarrow \infty} \inv{N} \bw^T \bD \bb_{\tau} \overset{a.s.}{=} 0 \label{eq:w_Aq_independence}
\end{gather}

\noindent
for any $\tau, k \leq t$ and with $\bq_0 = - \bx$ and for any matrix $\bD \in \mathbb{R}^{M \times N}$ whose limiting spectral distribution has finite support. 

Lastly, for AMP, VAMP, MF-OAMP and CG-VAMP, it was shown that there exists a 1D State Evolution (SE) that defines the dynamics of the magnitude of the error propagated in the SMP algorithms

\noindent
\begin{equation}
    v_{h_{t+1}} = SE_{t+1}(v_{h_t})
\end{equation}

\noindent
The form of the function $SE_{t+1}$ depends on the chosen SMP algorithm, but is independent of a particular realization of the true signal $\bx$ \cite{AMP_convergence_general_A}, \cite{NS-VAMP}, \cite{UnifiedSE}. A similar evolution can be defined for WS-CG-VAMP and CAMP, but it would be dependent on the whole set $\{ v_{h_{\tau}} \}_{\tau=0}^t$ instead of a single $v_{h_t}$ \cite{OurSecondPaper}, \cite{CAMP}. The SE provides the means of optimizing the functions $\bff_L$ and $\bff_D$ to obtain the optimal performance of the algorithm and provides a theoretical tool to study the stability and efficiency of SMP algorithms. In particular, the SE was used in \cite{SE_AMP}, \cite{VAMP}, \cite{CAMP} to show that AMP, VAMP and CAMP can achieve Bayes optimal reconstruction under Assumptions 1-3 given the denoiser $\bg_D$ is Bayes optimal and the subsampling factor $\delta$ is above a certain threshold.

\section{Efficient estimation of the divergence in SMP algorithms} \label{sec:div_est_introduction}

In SMP algorithms, one of the key ingredients that ensures stable, efficient and predictable dynamics is the Onsager correction that involves the divergence $\alpha_t$ of the denoiser $\bg_D(\br_t)$. In this section we develop two theoretical models for the divergence $\alpha_t$ in SMP and propose the associated estimators that can be computed using only the observed data in the algorithm and do not require additional executions of the denoiser. We begin with an intuition behind our methods and then move to the formal results.

\subsection{Intuition}

In the center of the developed techniques are the following parametrized denoiser and its oracle error

\noindent
\begin{gather}
    \bhf(\br_t, \halfa) = \bg_D(\br_t) - \halfa \br_t \label{eq:bhf} \\
    \bhq_{t+1}(\halfa) = \bhf(\br_t, \halfa) - \bx \label{eq:bhq}
\end{gather}

\noindent
where $\halfa$ is a scalar parameter. Note that when $\halfa = \alpha_t$, \eqref{eq:bhf} is an instance of \eqref{eq:f_d_VAMP} with $C_t = 1$ and therefore $\bhq_{t+1}(\alpha_t)$ follows the asymptotic identities \eqref{eq:h_t_gaussian} - \eqref{eq:w_Aq_independence}. However, here we stress that the parametrized denoiser \eqref{eq:bhf} is used only to prove certain steps, while in the algorithm we use either \eqref{eq:f_d_AMP} or \eqref{eq:f_d_VAMP}, depending on the chosen SMP algorithm.

The idea behind our method is to seek a function $E(\halfa)$ that has a root at $\alpha_t$ and we could solve for it. A straightforward example would be 

\noindent
\begin{equation}
    E(\halfa) = \inv{N} \bh_t^T \bhq_{t+1}(\halfa) \label{eq:naive_E}
\end{equation}

\noindent
Then, one could recover $\alpha_t$ by solving $E(\halfa) = 0$. However, this example of $E(\halfa)$ cannot be implemented in practice since it is explicitly formulated in terms of the error vectors that are not available. In this work we use the observed quantities in the algorithm to construct two types of practical functions $E(\halfa)$ that equated to zero can be used to estimate $\alpha_t$.

\subsection{Algebraic divergence estimator} \label{sec:algebraic_divergence_estimator}

The first class of estimators we propose is a practical extension of the naive and unavailable estimator \eqref{eq:naive_E}. To derive the method, rewrite $\bhq_{t+1} = \bhf(\br_t, \halfa) - \bx$ and $\bh_t = \br_t - \bx$ and consider the asymptotic regime $N \rightarrow \infty$ to obtain

\noindent
\begin{equation*}
    \lim_{N \rightarrow \infty} \inv{N} (\br_t - \bx)^T (\bhf(\br_t, \halfa) - \bx) \overset{a.s.}{=} \lim_{N \rightarrow \infty} \inv{N} (\br_t - \bx)^T \bhf(\br_t, \halfa) 
\end{equation*}

\noindent
where we used \eqref{eq:h_q_independence} to show that  $\lim_{N \rightarrow \infty} \inv{N} (\br_t - \bx)^T \bx \overset{a.s.}{=} 0$. Still, the above equation involves $\bx$ explicitly, which can be resolved by considering the difference $\br_t - \br_{t-1}$ instead of $\br_t$ alone

\noindent
\begin{align}
    &\lim_{N \rightarrow \infty} \inv{N} (\br_t - \br_{t-1})^T (\bhf(\br_t, \halfa) - \bx) \nonumber\\
    &= \lim_{N \rightarrow \infty} \inv{N} (\bh_t + \bx - \bh_{t-1} - \bx)^T (\bhf(\br_t, \halfa) - \bx) \nonumber\\
    &= \lim_{N \rightarrow \infty} \inv{N} (\bh_t - \bh_{t-1})^T (\bhf(\br_t, \halfa) - \bx) \nonumber\\
    &\overset{a.s.}{=} \lim_{N \rightarrow \infty} \inv{N} (\bh_t - \bh_{t-1})^T \bhf(\br_t, \halfa) \label{eq:first_E}
\end{align}

\noindent
This result suggests that if we define a scalar function

\noindent
\begin{equation}
    E_1(\halfa) = (\br_t - \br_{t-1})^T \bhf(\br_t, \halfa) \label{eq:E_1}
\end{equation}

\noindent
then we can recover such $\halfa$ that ensures the orthogonality between $\bhq_{t+1}(\halfa)$ and $\bh_t$ and $\bh_{t-1}$. The following theorem summarizes and generalizes this idea.

\begin{theorem} \label{th:closed_form_estimator}
    
    Given an SMP algorithm \eqref{eq:LB}-\eqref{eq:DB} with a denoiser $\bg_D(\br_t)$, under Assumptions 1-3 and assuming \eqref{eq:h_t_gaussian} - \eqref{eq:w_Aq_independence} hold up to iteration $t$, the divergence $\alpha_t$ of $\bg_D(\br_t)$ almost surely converges to
    
    \noindent
    \begin{equation}
        \lim_{N \rightarrow \infty} \alpha_t \overset{a.s.}{=} \frac{(\br_t - \bbr_t)^T \bg_D(\br_t)}{(\br_t - \bbr_t)^T \br_t} \label{eq:closed_form_divergence_estimator}
    \end{equation}
    
    \noindent
    when 
    
    \noindent
    \begin{equation}
        \lim_{N \rightarrow \infty} (\bh_t - \bbh_t)^T \bh_t \overset{a.s.}{\neq} 0
    \end{equation}
    
    \noindent
    where
    
    \noindent
    \begin{equation}
        \bbr_t = \sum_{\tau = 0}^{t-1} k_{\tau}^t \br_{\tau} \quad \quad \bbh_t = \sum_{\tau = 0}^{t-1} k_{\tau}^t \bh_{\tau} \label{eq:bbr_bbh}
    \end{equation}
    
    \noindent
    with scalar weights $\sum_{\tau = 0}^{t-1} k_{\tau}^t = 1$.
    
\end{theorem}

\begin{proof}
    First, due to the normalization of the weights $k_{\tau}^t$, we have that $\bbr_t - \bx = \bbh_t$. Then, consider the inner-product

    \noindent
    \begin{equation}
        \inv{N} (\br_t - \bbr_t)^T \bg_D(\br_t) = \inv{N} (\bh_t - \bbh_t)^T \bg_D(\br_t) \label{eq:r_br_g}
    \end{equation}
    
    \noindent
    where $\bbr_t$ and $\bbh_t$ are as in \eqref{eq:bbr_bbh}. Next, using the Strong Low of Large Numbers, the Stein's Lemma \cite{SURE} and the definition of the error vectors $\bh_{\tau}$ from \eqref{eq:h_t_gaussian}, we can show that \eqref{eq:r_br_g} almost surely converges to
    
    \noindent
    \begin{equation*}
        \lim_{N \rightarrow \infty} \inv{N} (\bh_t - \bbh_t)^T \bg_D(\br_t) \overset{a.s.}{=} \lim_{N \rightarrow \infty} \alpha_t \inv{N} (\bh_t - \bbh_t)^T \bh_t 
    \end{equation*}
    
    \noindent
    Next, using \eqref{eq:h_q_independence}, we can show that
    
    \noindent
    \begin{equation}
        \lim_{N \rightarrow \infty} \inv{N} (\br - \bbr_t)^T \br_t \overset{a.s.}{=} \lim_{N \rightarrow \infty} \inv{N} (\bh - \bbh_t)^T \bh_t \label{theorem_denominator_identity}
    \end{equation}
    
    \noindent
    Then, assuming $\lim_{N \rightarrow \infty} \inv{N} (\bh_t - \bbh_t)^T \bh_t \overset{a.s.}{\neq} 0$, we can arrive at the identity for $\alpha_t$ as in \eqref{eq:closed_form_divergence_estimator} by dividing both sides in the last result by $\inv{N} (\bh_t - \bbh_t)^T \bh_t$.
\end{proof}

In the following, we refer the estimator based on \eqref{eq:closed_form_divergence_estimator} as an \textit{algebraic estimator}.

By equating \eqref{eq:E_1} to zero and solving for $\halfa$, one can show that the indication function $E_1$ leads to the algebraic estimator with $\bbr_t = \br_{t-1}$. While in the LSL \eqref{eq:closed_form_divergence_estimator} holds for any set of weights $k_{\tau}^t$ as long as the normalization is satisfied, in the finite dimensional case these weights are expected to affect the accuracy of estimation. When $N$ is finite, the asymptotic identities used to derive Theorem \ref{th:closed_form_estimator} are no longer exact and an additional error emerges. This error might be substantial in the case if, for example, we use $\bbr_t = \br_{t-1}$. In this case, the term $\inv{N} (\bh - \bbh_t)^T \bx$ assumed to be equal to zero in \eqref{theorem_denominator_identity} might have considerable magnitude due to the fact that the magnitude of $\bx$ remains the same throughout the algorithm and might significantly exceed the magnitude of $\bh_t$ and of $\bbh_t$. Then, any small alignment of these error vectors with $\bx$ would result in a substantial quantity that affects the accuracy of the LSL approximation \eqref{theorem_denominator_identity}. 

On the other hand, the finite dimensionality also implies that the asymptotic evolution model of $\bh_t$ is corrupted by error that accumulates as the algorithm progresses. One of the effects of this error is that the core identity used to prove Theorem \ref{th:closed_form_estimator}

\noindent
\begin{equation}
    \inv{N} \bh_{\tau}^T \bg_D(\br_t) = \alpha_t \inv{N} \bh_{\tau}^T \bh_t
\end{equation}

\noindent
becomes less accurate for finite $N$ as the difference between $t$ and $\tau$ increases. For this reason we might observe poor quality of divergence estimates if we use $\bbr_t = \br_0$. The detailed analysis of the optimal choice of weights $k_{\tau}^t$ is left for further study, while in this work we consider the cases $\bbr_t = \br_{t-1}$ and $\bbr_t = \br_0$. The important advantage of these two options is that the computational cost of the resulting algebraic estimator is dominated by the cost of two inner-products of $N$-dimensional vectors. Such a low cost allows one to efficiently tune the denoiser using the SURE technique \cite{SURE} to optimize the performance of the denoising block. Yet, as it will be demonstrated in the simulation section, theses types of the algebraic estimator are sensitive to finiteness of $N$ and $M$, and demonstrate satisfactory accuracy only for inverse problems of dimension of order $10^6$ and larger. Next, we present another estimation method that demonstrates much higher robustness at lower dimensions, discuss the pros and cons of both methods and suggest a strategy to leverage the advantages of both of them.

\subsection{Polynomial divergence estimator} \label{sec:poly_div_est}

In this section we present another way of constructing a practical indication function $E(\halfa)$ that has a zero at the divergence $\alpha_t$ of the denoiser $\bg_D(\br_t)$. To obtain such a function, consider a Stein's Unbiased Risk Estimate (SURE) \cite{SURE} for the parametrized denoiser $\bhf(\br_t,\halfa)$. Using \eqref{eq:bhf}, the Strong Law of Large Numbers and Stein's Lemma \cite{SURE}, we can show that the MSE $\lim_{N \rightarrow \infty} \inv{N} ||\bhf(\br_t, \halfa) - x||^2$ almost surely converges to

\noindent
\begin{align}
    \lim_{N \rightarrow \infty} &\inv{N} ||\bhf(\br_t, \halfa) - \bx||^2 \overset{a.s.}{=} \lim_{N \rightarrow \infty} \inv{N} ||\bhf(\br_t, \halfa)||^2 + v_x \nonumber\\
    &- \lim_{N \rightarrow \infty} 2 \inv{N} \br_t^T \bhf(\br_t, \halfa) + 2 (\alpha_t - \halfa) v_{h_t} \label{eq:bhf_SURE}
\end{align}

\noindent
where we defined $v_x \overset{a.s.}{=} \lim_{N \rightarrow \infty} ||\bx||^2$ to be the variance of $\bx$. Then, if we define a function 

\noindent
\begin{equation}
    J_1(\halfa) = \inv{N} ||\bhf(\br_t, \halfa)||^2 - 2 \inv{N} \br_t^T \bhf(\br_t, \halfa) + v_x
\end{equation}

\noindent
we have that $J_1(\alpha_t)$ is an unbiased estimator of the MSE for $\bhf(\br_t, \alpha_t)$ and the term $\eta_{t+1} = 2 (\alpha_t - \halfa) v_{h_t}$ is the risk's estimation error that arises when we choose $\halfa \neq \alpha_t$. At the same time, the same MSE can be obtained from a different estimator that leverages the SE identity \eqref{eq:b_t_gaussian} and \eqref{eq:w_Aq_independence}. Using the definition of the vectors $\by$ and $\bhq_{t+1}(\halfa)$, and the Strong Law of Large Numbers, we can show that

\noindent
\begin{align}
    &\lim_{N \rightarrow \infty} \inv{N} ||\by - \bA \bhf(\br_t,\halfa)||^2 = \lim_{N \rightarrow \infty} \inv{N} ||\bw - \bA \bhq_{t+1}(\halfa)||^2 \nonumber\\
    &\overset{a.s.}{=} \delta v_w + \lim_{N \rightarrow \infty} \inv{N} ||\bA \bhq_{t+1}(\halfa)||^2 - \frac{2}{N} \bw^T \bA \bhq_{t+1}(\halfa) \label{eq:SE_divergence_criterion}
\end{align}

\noindent
Here we can use the conditioning technique \cite{SE_AMP}, \cite{VAMP}, \cite{UnifiedSE} for the random matrix $\bA$ to study the interaction between $\bhq_{t+1}(\halfa)$ and $\bA$. In Appendix A we show that the first inner-product in \eqref{eq:SE_divergence_criterion} corresponds to

\noindent
\begin{equation*}
    \lim_{N \rightarrow \infty} \inv{N} ||\bA \bhq_{t+1}(\halfa)||^2 \overset{a.s.}{=} \lim_{N \rightarrow \infty} \inv{N} ||\bhq_{t+1}(\halfa)||^2 + \zeta_{t+1}(\halfa)
\end{equation*}

\noindent
where $\zeta_{t+1}(\halfa)$ depends on the whole history of vectors $(\bh_t, \bh_{t-1}, ..., \bh_0)$ and $(\bq_t, \bq_{t-1}, ..., \bq_0)$ when $\halfa \neq \alpha_t$ and almost surely converges to zero for $\halfa = \alpha_t$. Similarly, one can show that $\inv{N} \bw^T \bA \bhq_{t+1}(\halfa)$ almost surely converges to zero for $\halfa = \alpha_t$. Therefore one could define another MSE estimator

\noindent
\begin{equation}
    J_2(\halfa) = \inv{N} ||\by - \bA \bhf(\br_t,\halfa)||^2 - \delta v_w
\end{equation}

\noindent
that coincides with the exact MSE $\lim_{N \rightarrow \infty} \inv{N} ||\bhf(\br_t, \halfa) - x||^2$ for $\halfa = \alpha_t$, but involves another risk's estimation error $\beta_{t+1} = \zeta_{t+1}(\halfa) - 2 \inv{N} \bw^T \bA \bhq_{t+1}(\halfa)$ for $\halfa \neq \alpha_t$. 

The important observation about $J_1(\halfa)$ and $J_2(\halfa)$ is that their errors $\eta_{t+1}$ and $\beta_{t+1}$ behave differently for $\halfa \neq \alpha_t$ and both almost surely converge to zero for $\halfa = \alpha_t$. This implies that we could recover $\alpha_t$ by finding the appropriate root to

\noindent
\begin{equation}
    E_2(\halfa) = J_1(\halfa) - J_2(\halfa) = 0 \label{eq:E_1_minus_E_2}
\end{equation}

\noindent
The following theorem shows that \eqref{eq:E_1_minus_E_2} corresponds to a particular quadratic equation.

\begin{theorem} \label{th:two_quadratic_equations}
    Consider an example of SMP algorithms \eqref{eq:LB}-\eqref{eq:DB} with a denoiser $\bg_D(\br_t)$ that takes as an input the vector $\br_t = \bx + \bh_t$ with $\bh_t \sim \normDensity(\bzero, v_{h_t} \bI_N)$. Then, the equation \eqref{eq:E_1_minus_E_2} corresponds to the following quadratic equation
    
    \noindent
    \begin{gather}
        \lim_{N \rightarrow \infty} u_1 + u_2 \halfa + u_3 \halfa^2 = 0 \label{eq:quadratic_equation_1} 
    \end{gather}
    
    \noindent
    where scalar coefficients are defined as
    
    \noindent
    \begin{gather}
        u_1 = \inv{N} \Big( ||\bg_D - \br_t||^2 - v_{h_t} - \big( ||\bA \bg_D - \by||^2 - \delta v_w \big) \Big) \nonumber \\
        u_2 = \frac{2}{N} \Big( \big( \br_t - \bg_D \big)^T \br_t - \big( \by - \bA \bg_D \big)^T \bA \br_t \Big) \label{eq:u_2}\\
        u_3 = \inv{N} ||\br_t||^2 -\inv{N} ||\bA \br_t||^2 \label{eq:u_3}
    \end{gather}
    
    \noindent
    and $\bg_D$ is used as a shorthand for $\bg_D(\br_t)$. Additionally, under Assumptions 1-3 and assuming \eqref{eq:h_t_gaussian} - \eqref{eq:w_Aq_independence} hold up to iteration $t$, the divergence $\alpha_t$ of $\bg_D(\br_t)$ is a root to \eqref{eq:quadratic_equation_1}.
    
\end{theorem}

\begin{proof}
    See Appendix A
\end{proof}

\noindent
Note that the coefficients of the equation \eqref{eq:quadratic_equation_1} are formed only from the available data at iteration $t$. Then, one way to estimate $\alpha_t$ is by computing the roots to \eqref{eq:quadratic_equation_1} and identifying which of the two roots is the correct one. In the following, we refer to this estimator as a \textit{polynomial estimator}.

Given a method for identifying the right root, the computational cost of the polynomial estimator is dominated by the cost of computing two matrix-vector products $\bA \bg_D(\br_t)$ and $\bA \br_t$. This cost can be reduced by reusing the results to form the updated vector $\bz_{t+1}$ 

\noindent
\begin{equation*}
    \bz_{t+1} = \by - \bA \bs_{t+1} = \by - \bA \bg_D(\br_t) - \alpha_t \bA \br_t
\end{equation*}

\noindent
In that case, the per-iteration computation cost is dominated by one matrix-vector product with $\bA$. 

As it will be demonstrated in the simulation section, SMP algorithms with the polynomial estimator demonstrate stable dynamics similar to the BB-MC estimator even for $N$ and $M$ of order $10^4$. When the dimension increases, the accuracy of correction of the proposed method is even superior to BB-MC and leads to improved dynamics of the SMP algorithms. To combine the advantages of the algebraic and the polynomial estimators, we suggest the use of the algebraic estimator to tune the denoiser via SURE and use the polynomial estimator to compute the final correction scalar $\alpha_t$. This approach combines the advantages of both methods and results in a fast and efficient performance of the denoising block.

\subsection{Root identification for the polynomial estimator} \label{sec:tendency_in_root_identification}

While Theorem \ref{th:two_quadratic_equations} relates the correction scalar $\alpha_t$ to one of the the roots $\halfa_1$ and $\halfa_2$ of \eqref{eq:quadratic_equation_1}, it is still required to identify which of the two roots is the right one. In this subsection, we study the LSL properties of the quadratic equation \eqref{eq:quadratic_equation_1} and propose a method for assigning $\alpha_t$ to either $\halfa_1$ or $\halfa_2$. 

To establish the theoretical connection between $\alpha_t$ and $\halfa_1$ and $\halfa_2$, we refer to the following two properties of a generic quadratic equation. We know that a quadratic polynomial corresponds to a parabola, which is oriented either up or down and this orientation is uniquely identified based on the sign of $u_3$. Second, if we know the orientation of the parabola, we can uniquely identify the root if we know the sign of the derivative of the polynomial at the desired root. Here, the derivative of the polynomial from \eqref{eq:quadratic_equation_1} corresponds to

\noindent
\begin{equation}
    \frac{\partial}{\partial \halfa} \Big( u_1 + u_2 \halfa + u_3 \halfa^2 \Big) = u_2 + 2 u_3 \halfa \label{eq:quadratic_equation_derivative}
\end{equation}

\noindent
In this case, if we worked out the sign of \eqref{eq:quadratic_equation_derivative} at $\halfa = \alpha_t$ and found out that it is, for example, always positive, then $\alpha_t$ would correspond to the smallest root of \eqref{eq:quadratic_equation_1} if $u_3$ is positive and to the largest root if $u_3$ is negative. The opposite would hold if it turns out that \eqref{eq:quadratic_equation_derivative} at $\halfa = \alpha_t$ is always negative. The following theorem presents the LSL identities for \eqref{eq:quadratic_equation_derivative} at $\halfa = \alpha_t$ for MF-OAMP, VAMP and CG-VAMP.

\begin{theorem} \label{th:divergence_identity_OAMP_based_algorithms}
    
    Consider the MF-OAMP, VAMP and CG-VAMP algorithms equipped with $\bg_D(\br_t)$ denoiser and let $\alpha_t$ be the divergence of $\bg_D(\br_t)$. Define a corrected denoiser
    
    \noindent
    \begin{equation}
        \bbf(\br_t) = \bg_D(\br_t) - \alpha_t \br_t \label{eq:th_DF_func}
    \end{equation}
    
    \noindent
    and its error
    
    \noindent
    \begin{equation*}
        \bbq_{t+1} = \bbf(\br_t) - \bx
    \end{equation*}
    
    \noindent
    Additionally define an inner-product $\psi_t = \inv{N} \bq_t^T \bbq_{t+1}$, where $\bq_t$ is the error $\bq_t = \bs_t - \bx$ from the previous iteration. Then, under Assumptions 1-3 and assuming \eqref{eq:h_t_gaussian} - \eqref{eq:w_Aq_independence} hold up to iteration $t$, the derivative \eqref{eq:quadratic_equation_derivative} at $\halfa = \alpha_t$ almost surely converges to
    
    \begin{itemize}
    
        \item For MF-OAMP:
        
        \noindent
        \begin{equation}
            \lim_{N \rightarrow \infty} \frac{1}{2} \big( u_2 + 2 \alpha u_3 \big) \overset{a.s.}{=}  \big( 1 - \chi_2 \big) \big( \psi_t - v_{q_t} \big) \label{eq:LSL_deriv_MF_OAMP}
        \end{equation}
        
        \noindent
        where $\chi_2 = \inv{N} Tr \big\{ (\bA \bA^T)^2 \big\}$.
        
        \item For VAMP:
        
        \noindent
        \begin{equation}
            \lim_{N \rightarrow \infty} \frac{1}{2} \big( u_2 + 2 \alpha u_3 \big) \overset{a.s.}{=} \frac{\big(v_w - v_{h_t} \big)}{v_{q_t}} \big( \psi_t - v_{q_t} \big) \label{eq:LSL_deriv_VAMP}
        \end{equation}
        
        \item For CG-VAMP:
        
        % \noindent
        % \begin{align}
        %     &\lim_{N \rightarrow \infty} \big( u_2 + 2 \alpha u_3 \big) \overset{a.s.}{=} 2 \bigg[ \frac{2v_w - v_{h_t}}{v_{q_t}} \big(\psi_t - v_{q_t}\big) \nonumber\\
        %     &+ v_w + \frac{\frac{\psi_t}{v_{q_t}} \Big( v_w \inv{N} \bw^T \bmu_t^i - \delta v_w + (k_t - 1)\Big)}{\gamma_t v_{q_t}} \bigg] \label{eq:LSL_deriv_CG_VAMP}
        % \end{align}
        
        \noindent
        \begin{align}
            &\lim_{N \rightarrow \infty} \frac{1}{2} \big( u_2 + 2 u_3 \alpha \big) \overset{a.s.}{=} \frac{k_t v_{q_t} - \psi_t}{\gamma_t v_{q_t}} - (v_{q_t} - \psi_t) - v_w \nonumber\\
            &+ \frac{ \frac{\psi_t}{v_{q_t}} \Big( \lim_{N\rightarrow \infty} v_w \inv{N} \bw^T \bmu_t^i - \delta v_w + 2 v_w \gamma_t v_{q_t} \Big) }{ \gamma_t v_{q_t}} \label{eq:LSL_deriv_CG_VAMP}
        \end{align}
        
        \noindent
        where $\bmu_t^i$ is the CG approximation with $0 \leq i \leq M$ iterations of the system of linear equations \eqref{eq:SLE} and $k_t = \frac{\inv{N}||\bA^T \bmu_t^i||^2}{\gamma_t}$.
        
    \end{itemize}
    
\end{theorem}

\begin{proof}
    For the proof of \eqref{eq:LSL_deriv_MF_OAMP} and \eqref{eq:LSL_deriv_VAMP} see Appendix B. The proof of \eqref{eq:LSL_deriv_CG_VAMP} is rather technical and is omitted in the current work for the sake of space. The proof of \eqref{eq:LSL_deriv_CG_VAMP} and the additional analysis of this asymptotic result is available in the supplementary materials.
\end{proof}

\noindent
Before proceeding next, we would like to emphasize that the function $\bbf(\br_t)$ uses the proper correction scalar $\alpha_t$ in the contrast to $\bhf(\br_t)$, which uses an arbitrary parameter $\halfa$. At the same time, $\bbf(\br_t)$ naturally emerges in the proof of the theorem, while in the actual algorithm we use the corrected denoiser $\bff_D(\br_t)$ from \eqref{eq:f_d_VAMP}, which assumes an additional scaling by an arbitrary scalar $C_t$.

Equipped with Theorem \ref{th:divergence_identity_OAMP_based_algorithms}, our goal is to identify whether the derivatives \eqref{eq:LSL_deriv_MF_OAMP}-\eqref{eq:LSL_deriv_CG_VAMP} are positive or negative and we start with assuming that the common factor $\psi_t - v_{q_t}$ is negative

\noindent
\begin{equation}
    \psi_t - v_{q_t} < 0 \label{eq:psi_less_v_q}
\end{equation}

\noindent
We motivate this assumption in the following way. First, note that to show \eqref{eq:psi_less_v_q} holds, it is sufficient to show that $\inv{N} ||\bbq_{t+1}||^2$ is less or equal to $\inv{N} ||\bq_t||^2$. Doing this for a general denoiser and for a general free parameter $C_t$ in the update \eqref{eq:f_d_VAMP} is challenging and here we assume that the denoiser is close to the Bayes-optimal denoiser. Then, from \cite{VAMP}, \cite{OAMP}, \cite{EP_Keigo} we know that when the denoiser $\bg_D$ is Bayes-optimal, the variance $v_{q_{t+1}}$ is monotonic with $v_{h_t}$. Therefore, next we assume that our SMP algorithm is progressing after iteration $t$ so that the magnitude of the intrinsic noise is decreasing 

\noindent
\begin{equation}
    v_{h_t} < v_{h_{t-1}} \label{eq:v_h_progressing}
\end{equation}

\noindent
This would justify the assumption that $\inv{N} \bq_{t+1}^T \bq_t - \inv{N} \bq_t^T \bq_t$ is negative, but in our case, the vector $\bbq_{t+1}$ assumes $\overline{C}_{t+1} = 1$ in the update \eqref{eq:f_d_VAMP}, while the error $\bq_{t}$ might be using the optimal value of $C_t$ \cite{OAMP}, \cite{VAMP}, \cite{EP_Keigo}

\noindent
\begin{equation}
    C_{t}^{optimal} = \frac{v_{h_{t-1}}}{v_{h_{t-1}} - MSE(\bg_D(\br_{t-1}))}
\end{equation}

\noindent
When we choose a powerful denoiser like BM3D that is close to the Bayes-optimal denoiser, the ratio $\frac{MSE(\bg_D(\br_t)}{v_{h_t}}$ usually is small, which leads to $C_{t}^{optimal}$ being close to 1. This implies that the update \eqref{eq:th_DF_func} is close to the optimal one for the Bayes-optimal denoisers and we expect to see that $\inv{N} ||\bbq_{t+1}||^2$ is less than $\inv{N} ||\bq_t||^2$ even for small improvements in $v_{h_t}$. This, together with \eqref{eq:v_h_progressing} motivates the assumption \eqref{eq:psi_less_v_q}.

Next we consider the other factors in \eqref{eq:LSL_deriv_MF_OAMP}-\eqref{eq:LSL_deriv_CG_VAMP}. Define $\chi_j$ to be the $j$-th emperical moment of the eigenvalues $\bS \bS^T$ as

\noindent
\begin{equation}
    \chi_j = \inv{N} Tr\Big\{ \big( \bA \bA^T \big)^j \Big\} \label{eq:chi}
\end{equation}

\noindent
Recall form Section \ref{sec:error_dynamics} that we assume the normalization $\chi_1 = 1$. Then, from the standard relationship of the second and the first moments we have that

\noindent
\begin{equation}
    \chi_2 = \chi_1^2 + \inv{N} Tr \big\{ ( \Lambda - \chi_1 )^2 \big\} \geq \chi_1^2 = 1 \label{eq:MF_OAMP_chi2}
\end{equation}

\noindent
Using \eqref{eq:psi_less_v_q} and \eqref{eq:MF_OAMP_chi2} in \eqref{eq:LSL_deriv_MF_OAMP}, we can show for MF-OAMP that the derivative \eqref{eq:quadratic_equation_derivative} is always positive. 

For VAMP, the identification of the sign of \eqref{eq:LSL_deriv_VAMP} can be completed using the following lemma

\begin{lemma} \label{lemma:v_h_minus_v_w_positivity}
    
    Consider the VAMP algorithm. Under Assumptions 1-3 and assuming \eqref{eq:h_t_gaussian} - \eqref{eq:w_Aq_independence} hold up to iteration $t$, the difference $v_{h_t} - v_w$ is positive if the following inequality holds
    
    \noindent
    \begin{equation}
        v_{q_t} > \frac{v_w}{\delta^{-1} - 1} \label{eq:v_h_t_minus_v_w_positivity_condition}
    \end{equation}
    
\end{lemma}

\begin{proof}
    See Appendix C
\end{proof}

Then, assuming the denoiser $\bg_D(\br_t)$ is Bayes-optimal, we guarantee $v_{h_t} - v_w > 0$ up till the iteration where $MSE(\bg_D(\br_t))$ drops below the right hand side of \eqref{eq:v_h_t_minus_v_w_positivity_condition}. Because $\delta = \frac{M}{N}$ tends to be much smaller than $1$, in practice, when the signal $\bx$ has complex structure like in the case with natural images, the denoiser's MSE is unlikely to drop below that level. Thus, we assume \eqref{eq:v_h_t_minus_v_w_positivity_condition} does hold and, together with \eqref{eq:psi_less_v_q}, this implies that \eqref{eq:LSL_deriv_VAMP} is positive.

To make a similar conclusion about positivity of \eqref{eq:LSL_deriv_CG_VAMP} for CG-VAMP, we note that both MF-OAMP and VAMP are special cases of CG-VAMP with $0$ and $M$ iterations for the CG algorithm respectively. Therefore, when the system of linear equations \eqref{eq:SLE} is poorly or accurately approximated, we already know that \eqref{eq:LSL_deriv_CG_VAMP} is positive. The analysis of the case, where the SLE \eqref{eq:SLE} is approximated with a moderate accuracy, is omitted in this work for the sake of space and is provided in the supplementary materials. Thus, in the following we assume that \eqref{eq:LSL_deriv_CG_VAMP} is positive when the SLE \eqref{eq:SLE} is approximated with a moderate accuracy.

Based on the above conclusions and assumptions, we propose the following strategy for identifying $\alpha_t$ among $\halfa_1$ and $\halfa_2$ for MF-OAMP, VAMP and CG-VAMP.

\begin{proposal} \label{proposal:OAMP_alpha_estimator}
    
    Consider the MF-OAMP, VAMP and CG-VAMP algorithms equipped with a denoiser $\bg_D(\br_t)$. Then, the divergence $\alpha_t$ of $\bg_D$ at $\br_t$ could be estimated as
    
    % Consider the MF-OAMP, VAMP and CG-VAMP algorithms with a denoiser $\bg_D(\br_t)$ and let $\alpha_t$ be the divergence of $\bg_D(\br_t)$. Then, we suggest to chose the correction scalar $\talfa_t$ as
    
    \noindent
    \begin{equation}
        \talfa_t = 
        \begin{cases} 
            max(\halfa_1,\halfa_2), & \mbox{if } u_3 < 0 \\ 
            min(\halfa_1,\halfa_2), & \mbox{if } u_3 > 0
        \end{cases} \label{eq:proposed_div_est_1}
    \end{equation}
    
    \noindent
    where the scalar $u_3$ is computed from \eqref{eq:u_3}.
    
\end{proposal}

\noindent
As it will be demonstrated in the simulations section, the proposed method always correctly identified $\alpha_t$ among $\halfa_1$ and $\halfa_2$ for MF-OAMP and CG-VAMP with the BM3D denoiser.

A similar result could be derived for AMP and CAMP, but the derivation is much more involved due to the complex recursive structure of the vector $\bh_t$ \cite{UnifiedSE}, \cite{CAMP} and is left for further work. Nevertheless, when $\bA$ is an i.i.d. Gaussian matrix, the dynamics of AMP and MF-OAMP are very similar and, due to the similarity of the algorithms, one could assume that the behaviour of the derivative \eqref{eq:quadratic_equation_derivative} is similar for both algorithms. In the simulation section we compare the performance of MF-OAMP and AMP and confirm that AMP with the BM3D denoiser and the polynomial estimator \eqref{eq:proposed_div_est_1} leads to stable dynamics similar to AMP with the BB-MC divergence estimator. A similar situation is observed with respect to WS-CG-VAMP, which has a complex recursive evolution model of the vector $\bh_t$ \cite{OurSecondPaper}, and for which proving the LSL result of the derivative \eqref{eq:quadratic_equation_derivative} is challenging. Yet, in Section \ref{sec:simulations_CG_VAMP_and_WS_CG_VAMP} we numerically confirm that WS-CG-VAMP with the estimator \eqref{eq:proposed_div_est_1} exhibits stable dynamics and even demonstrates a faster convergence rate compared to the same algorithm but with the BB-MC divergence estimator.

\subsection{Implementation details of the polynomial estimator}

As discussed in Section \ref{sec:algebraic_divergence_estimator}, when we consider a practical reconstruction problem of a finite dimensional signal $\bx$, there are additional stochastic components emerging in the algorithm. At the same time, certain sharp estimators might be sensitive to such error and provide inconsistent estimates. For example, when an SMP algorithm is close to a fixed point, we have observed that the roots to the quadratic equation \eqref{eq:quadratic_equation_1} might become complex. Since the target value -- the divergence of the denoiser $\bg_D(\br_t)$ with a real input $\br_t$ -- is real, we suggest to use the stationary point

\noindent
\begin{equation}
    \Tilde{\alpha}_t = -\frac{u_2}{2 u_3} \label{eq:alpha_estimator_at_fixed_point}
\end{equation}

\noindent
of the quadratic equation \eqref{eq:quadratic_equation_1} as the estimate of the correction scalar $\alpha_t$.

\section{Simulation results} \label{sec:simulation_beginning}

In this section we compare the proposed divergence estimators against the BB-MC method \cite{MC-divergence} within AMP, MF-OAMP, CG-VAMP and WS-CG-VAMP. We did not consider the VAMP algorithm explicitly because it requires precomputing the SVD of $\bA$, which is not feasible for large dimensions $N$ and $M$, while in this work we focus on large dimensional inverse problems. 

To the best of our knowledge, there is no general practice for tuning the BB-MC divergence estimator for denoisers that violate the continuity assumption like in the case of the BM3D denoiser. In this work we use the heuristic for choosing the scalar $\epsilon$ from \eqref{eq:black_box_div_estimator} as in the GAMP library\footnote{The link to the code is https://sourceforge.net/projects/gampmatlab/}

\noindent
\begin{equation*}
    \epsilon = 0.1 \min \big( \sqrt{v_{h_t}}, \inv{N} ||\br_t||_1 \big) + e
\end{equation*}

\noindent
where $e$ is the the float point precision in MATLAB. This choice  of  the  parameter $\epsilon$ demonstrated  stable  estimation throughout iterations $t$ for all the considered algorithms. 

Additionally, when we refer to the polynomial method for divergence estimation, we mean \eqref{eq:proposed_div_est_1} when the roots are real and \eqref{eq:alpha_estimator_at_fixed_point} for complex roots. Lastly, in all the experiments where BB-MC is involved, we use a single MC trial (additional execution of the denoiser) to estimate the divergence.

\subsection{Polynomial vs algebraic estimators}

We begin with the comparison of the polynomial estimator \eqref{eq:proposed_div_est_1} against the algebraic estimators \eqref{eq:closed_form_divergence_estimator} with $\bbr_t = \br_0$ and $\bbr_t = \br_{t-1}$. For this purpose, we consider the CG-VAMP algorithm, recovering a natural image demonstrated on the right of Figure \ref{fig:ground_truth_image} of dimension $2048$ by $2048$. We choose the measurement matrix $\bA$ to be the Fast ill-conditioned Johnson-Lindenstrauss Transform (FIJL)\cite{NS-VAMP} which acts as a prototypical ill-conditioned CS matrix. In our experiments, the FIJL operator $\bA = \bJ \bS \bP \bH \bD$ is composed of the following matrices \cite{NS-VAMP}: the values of the diagonal matrix $\bD$ are either $-1$ or $1$ with equal probability; the matrix $\bH$ is the Discreet Cosine Transform (DCT); The matrix $\bP$ is a random permutation matrix and the matrix $\bS$ is a diagonal matrix with geometric singular values that achieve the desired condition number as was considered in \cite{NS-VAMP}. Lastly, the subsampling matrix $\bJ$ is an $M$ by $N$ matrix that has ones on the main diagonal. For such a matrix, we set the condition number $\kappa(\bA) = 1000$ and the subsampling factor $\delta = 0.05$. Additionally, we set the measurement noise variance $v_w$ to achieve $40 dB$ SNR. Finally, we used BM3D\footnote{The BM3D library used throughout the simulations can be downloaded from the website of the authors of the denoiser http://www.cs.tut.fi/~foi/GCF-BM3D/. For this particular implementation we used the 'profile' to be 'np'.} denoiser \cite{BM3D} in the denoising block. 

\noindent
\begin{figure}
\centering
\includegraphics[width=0.5\textwidth]{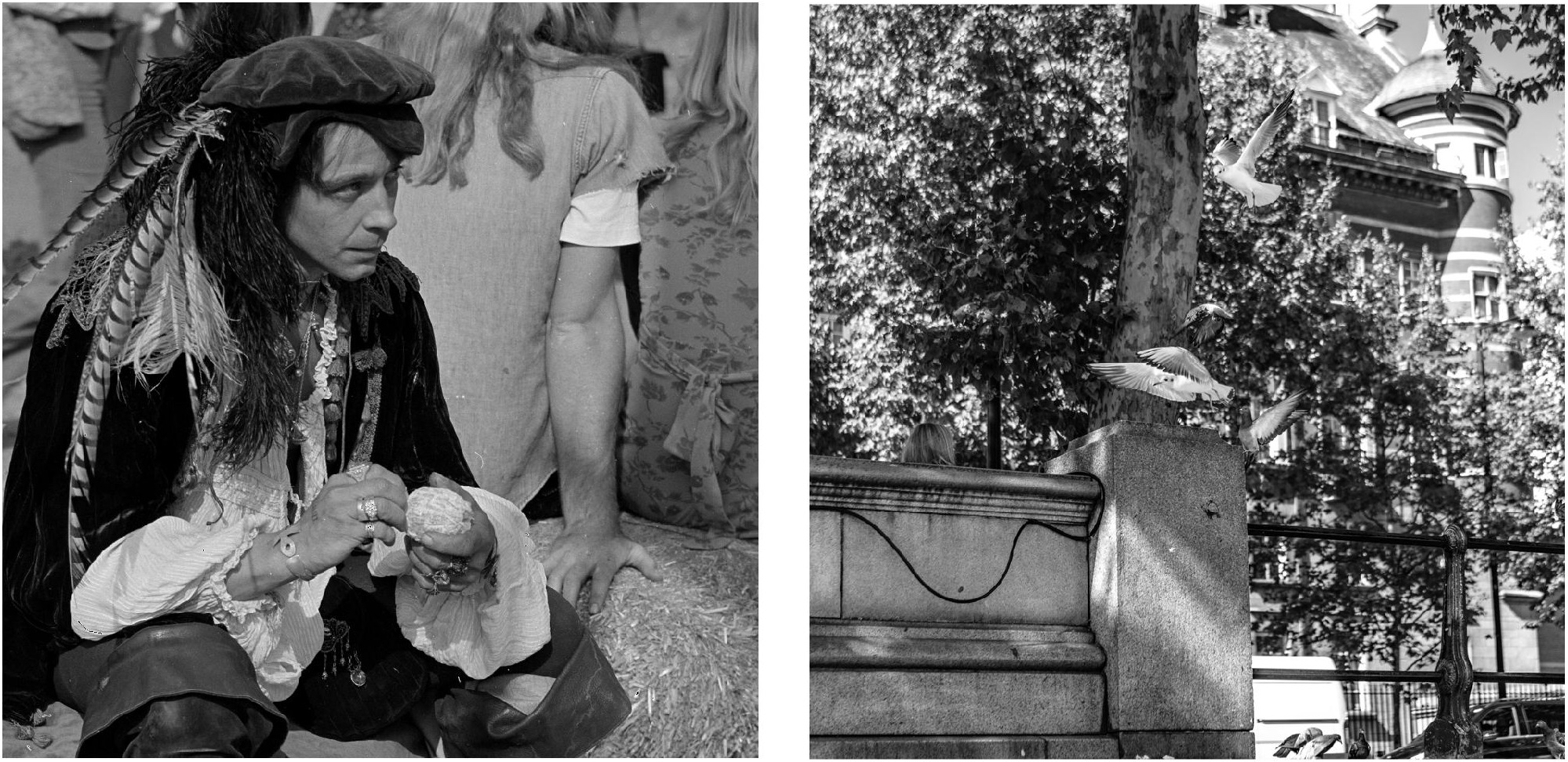}
\caption{The ground truth images}
\label{fig:ground_truth_image}
\end{figure}

In the first experiment, we run a single CG-VAMP algorithm where $\alpha_t$ is estimated by the polynomial estimator and, additionally, the two algebraic estimators are computed in parallel (these two values are not used within the algorithm and are only archived). For this experiment, we computed the normalized error $\frac{(\Hat{\alpha}_t - \alpha_t)^2}{(\alpha_t)^2}$, where $\Hat{\alpha}_t$ corresponds to either the estimate produced by the polynomial or by the two algebraic estimator, and the ``oracle'' correction $\alpha_t$ corresponds to

\noindent
\begin{equation}
    \alpha_t = \frac{\bh_t^T \bg_D(\br_t)}{N v_{h_t}} \label{eq:alpha_oracle}
\end{equation}

\noindent
The results averaged over $15$ iterations are shown on Figure \ref{fig:BM3D_CG_VAMP_2048}. As seen from the figure, the polynomial estimator demonstrates the best accuracy of estimating the ``oracle'' correction \eqref{eq:alpha_oracle}, while the algebraic estimator with $\bbr_t = \br_{t-1}$ demonstrates second to the best performance. On the other hand, the algebraic estimator with $\bbr_t = \br_{0}$ turns out to perform considerably worse than the other two and therefore is not recommended neither for computing $\alpha_t$ nor for estimating the divergence of the denoiser $\bg_D$ for its optimization via SURE. 

Next, we assess the stability of the algebraic estimator. For this, we compare two CG-VAMP algorithms: one where $\alpha_t$ is computed based on the polynomial estimator as in the previous experiment, and one where $\alpha_t$ is estimated by the algebraic estimator with $\bbr_t = \br_{t-1}$. Here, we computed the same error for $\alpha_t$ and the Normalized MSE (NMSE) $\frac{||\bg_D(\br_t) - \bx||^2}{||\bx||^2}$. The two error measures  averaged over $15$ realizations are shown on Figure \ref{fig:BM3D_CG_VAMP_2048_v_b_and_div}. As seen from the left plot depicting the NMSE, the CG-VAMP algorithm with the algebraic estimator with $\bbr_t = \br_{t-1}$ diverges the halfway through the execution, while the same algorithm but with the polynomial estimator demonstrates high stability. Based on this and the previous experiment, one could potentially use the algebraic estimator for computing a "rough" estimate of the divergence of $\bg_D$ to, for example, estimate the SURE, and use the polynomial estimator for computing the final correction $\halfa_t$ to ensure stable dynamics of SMP.

\noindent
\begin{figure}
\includegraphics[width=0.5\textwidth]{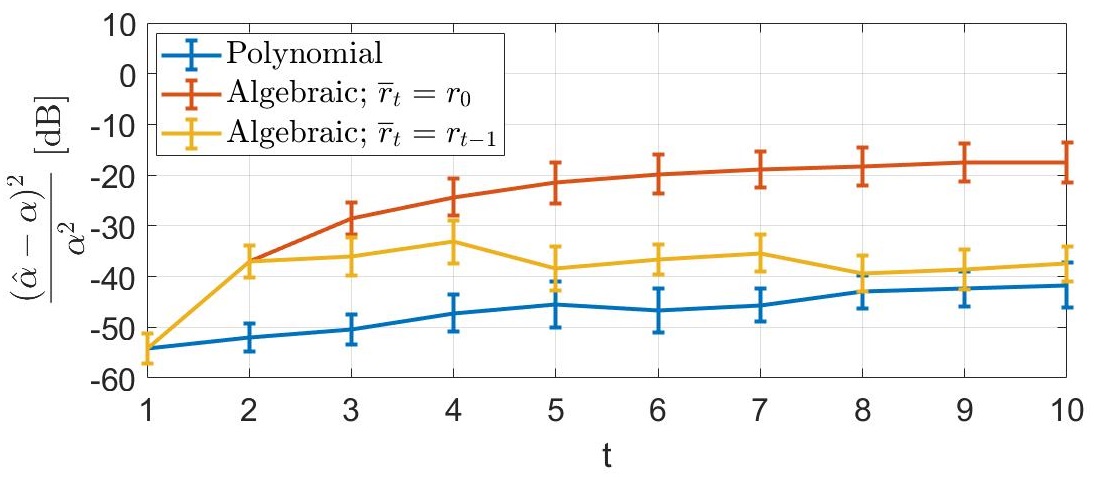}
\caption{Divergence estimation error with the standard deviation error bars of the polynomial and the algebraic estimators}
\label{fig:BM3D_CG_VAMP_2048}
\end{figure}

\noindent
\begin{figure}
\includegraphics[width=0.5\textwidth]{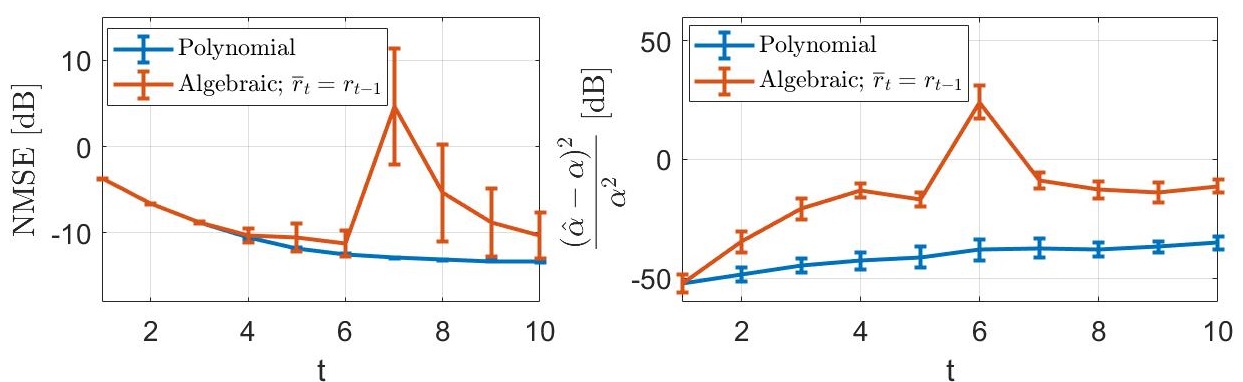}
\caption{Left: the Normalized MSE (NMSE). Right: divergence estimation error with the standard deviation error bars of the polynomial and the algebraic estimators}
\label{fig:BM3D_CG_VAMP_2048_v_b_and_div}
\end{figure}

\subsection{AMP and MF-OAMP}

Next, we compare the BB-MC and the polynomial estimators for estimating $\alpha_t$ within AMP and MF-OAMP algorithms. In particular, we consider the problem of recovering a natural image 'man' shown on the left of Figure \ref{fig:ground_truth_image} of dimension $400$ by $400$ from the set of measurements \eqref{eq:y_measurements} with subsampling factor $\delta = 0.05$. We chose $\bA$ to be a random Gaussian matrix with the normalization $\inv{N} Tr\{ \bA \bA^T \} = 1$ and chose the variance $v_w$ of the measurement noise $\bw$ that achieves $40dB$ SNR. Similarly, we used BM3D in the denoising block. In this setting, we compare the NMSE of two different AMP algorithms: one with BB-MC estimator and one with the polynomial estimator. The same experiment is repeated for MF-OAMP. The results averaged over $50$ realizations are demonstrated on Figure \ref{fig:bm3d_mf_oamp_amp_NMSE}. First, we notice that the dynamics of AMP and of MF-OAMP with the BB-MC divergence estimators are very similar. This observation supports the assumption in Section \ref{sec:tendency_in_root_identification} that the behaviour of the derivative \eqref{eq:quadratic_equation_derivative} should be very similar for the two algorithms as well and, therefore, supports the idea to use the polynomial estimator for AMP. This assumption is confirmed on the Figure \ref{fig:bm3d_mf_oamp_amp_NMSE} showing that both MF-OAMP and AMP with the two types of divergence estimators have almost identical reconstruction dynamics, while the execution time of the algorithms with the polynomial estimator was almost twice smaller. 

\noindent
\begin{figure}
\includegraphics[width=0.5\textwidth]{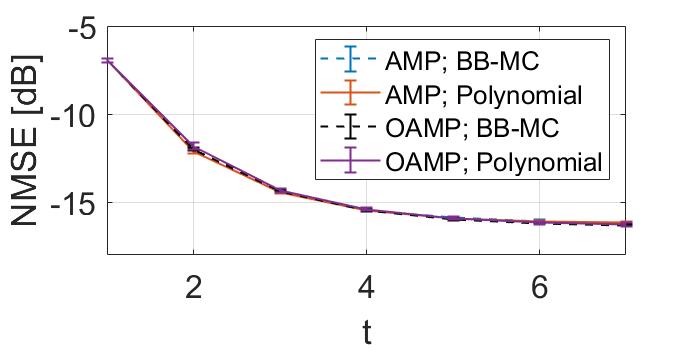}
\caption{NMSE with the standard deviation error bars of AMP and of MF-OAMP with the two divergence estimators: BB-MC \eqref{eq:black_box_div_estimator} and the polynomial divergence estimator \eqref{eq:proposed_div_est_1}.}
\label{fig:bm3d_mf_oamp_amp_NMSE}
\end{figure}

\subsection{CG-VAMP and WS-CG-VAMP} \label{sec:simulations_CG_VAMP_and_WS_CG_VAMP}

Next, we study the CG-VAMP and WS-CG-VAMP algorithms for the case where $\bA$ is a FIJL operator as in the first experiment. We consider recovering the same image 'man' at an increased resolution of $1024$ by $1024$, keeping the subsampling factor $\delta = 0.05$ and the measurement noise variance $v_w$ to give $40 dB$ SNR. We keep the BM3D denoiser and test the CG-VAMP algorithms for three condition numbers $\kappa(\bA) = (100, 1000, 10000)$. For all the executions we used the fixed number of iterations for the CG algorithm $i=5$. The NMSE of the algorithms averaged over $15$ realizations is shown on Figure \ref{fig:NLM_CG_VAMP_NMSE_different_delta}. As we see from the plot, the CG-VAMP algorithm with the polynomial divergence estimator demonstrates a similar reconstruction performance for smaller condition numbers $\kappa(\bA)$ and exhibits slightly improved convergence rate for larger $\kappa(\bA)$. Additionally we computed the estimation error of $\alpha$ as in the first experiment. The averaged result over 15 realizations for $\kappa(A) = 1000$ is depicted on Figure \ref{fig:NLM_CG_VAMP_alpha}. As seen from the plot, the polynomial estimator demonstrates higher accuracy of estimation for those iterations $t$, where the algorithm is not at the fixed point, and exhibits a similar accuracy to the BB-MC estimator when the algorithm converges.

\noindent
\begin{figure}
\includegraphics[width=0.5\textwidth]{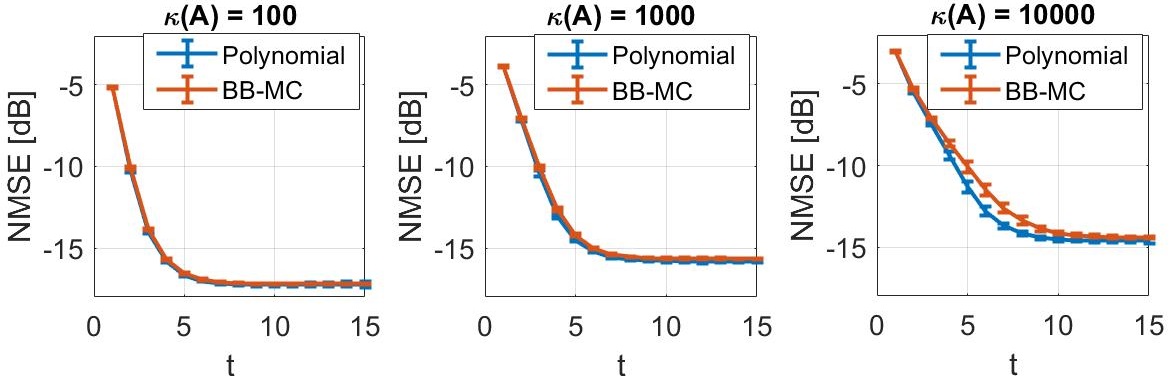}
\caption{NMSE with the standard deviation error bars for two CG-VAMP algorithms: with the BB-MC divergence estimator \eqref{eq:black_box_div_estimator} and with the polynomial divergence estimator \eqref{eq:proposed_div_est_1}.}
\label{fig:NLM_CG_VAMP_NMSE_different_delta}
\end{figure}

\noindent
\begin{figure}
\includegraphics[width=0.5\textwidth]{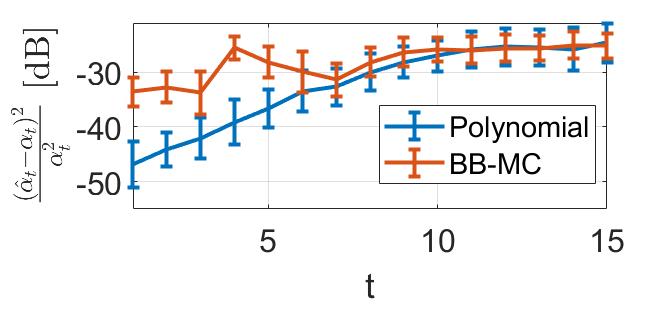}
\caption{NMSE with the standard deviation error bars of estimating the correction scalar $\alpha_t$ within CG-VAMP}
\label{fig:NLM_CG_VAMP_alpha}
\end{figure}

Next, we keep exactly the same setting as for CG-VAMP and consider the performance of WS-CG-VAMP with the BB-MC and the polynomial estimators. We use the practical version of the WS-CG-VAMP algorithm proposed in \cite{OurSecondPaper}. This version of WS-CG-VAMP has the same computational cost as the regular CG-VAMP algorithm, demonstrates stable and improved reconstruction properties, but violates the exact State Evolution. We compared the two versions of WS-CG-VAMP against CG-VAMP with the same number of iterations for CG $i=5$ and against ``pseudo-VAMP'', which corresponds to CG-VAMP with $i=500$. This number of inner-loop iterations achieves the relative residual $\frac{||\bz_t - \bW_t \bmu_t^{i}||^2}{||\bz_t||}$ of order $-70dB$, which justifies the usage of such an algorithm as the benchmark. The NMSE of the algorithms averaged over $15$ realizations is shown on Figure \ref{fig:NLM_WS_CG_VAMP_vs_CG_VAMP_NMSE}. Despite the fact that the considered version of WS-CG-VAMP violates the SE, which is required for the derivation of the polynomial divergence estimation method, the reconstruction quality of WS-CG-VAMP with the divergence estimator \eqref{eq:proposed_div_est_1} is similar to WS-CG-VAMP with the BB-MC divergence estimator. Similarly to \cite{OurSecondPaper}, in this experiment we observe that WS-CG-VAMP with only $5$ inner-loop iterations is able to converge to a fixed point with almost identical quality as the pseudo-VAMP algorithm's.

\noindent
\begin{figure}
\includegraphics[width=0.5\textwidth]{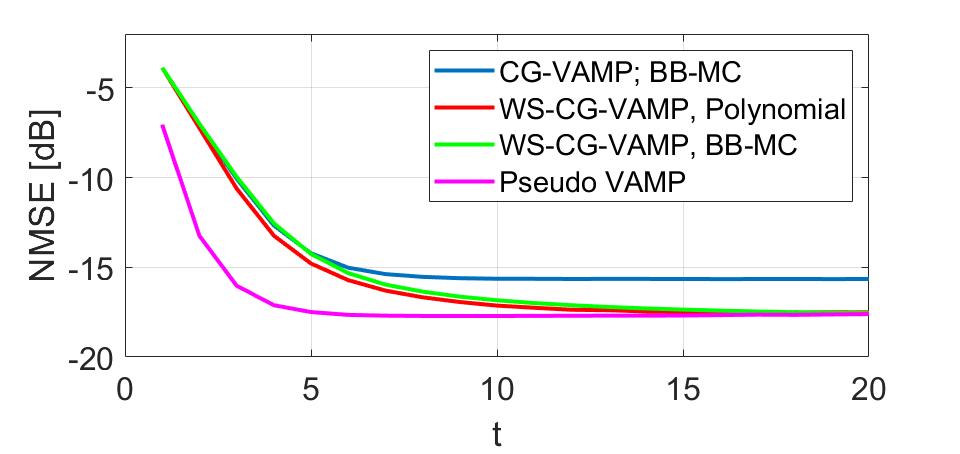}
\caption{The comparison of CG-VAMP (blue), WS-CG-VAMP with BB-MC estimator (green) and with the proposed estimator \eqref{eq:proposed_div_est_1} (red). The purple curve corresponds to the approximated VAMP algorithm achieved by setting the number of CG iterations to $i=500$.}
\label{fig:NLM_WS_CG_VAMP_vs_CG_VAMP_NMSE}
\end{figure}

\section{Conclusions}

In this work we have proposed two alternative to the traditional Black-Box Monte Carlo (BB-MC) \cite{MC-divergence} methods for estimating the divergence of denoisers within SMP algorithms. Similarly to BB-MC, the proposed methods do not use any additional information about the denoiser apart from its input and output. However, contrary to the BB-MC method, the two suggested estimators do not require executing the denoiser additional times and, therefore, significantly accelerate the SMP algorithm when an expensive denoiser such as BM3D is used. Due to the negligible computational cost of the algebraic method, we believe that it could be used to tune the denoiser via SURE \cite{SURE} to achieve the best performance of the denoising block in SMP. The second estimation method -- the polynomial estimator -- complements the first one and demonstrates high robustness with respect to the dimensionality of the inverse problem and improved accuracy of correction compared to the BB-MC method.

While the two proposed estimators are exact in the Large System Limit (LSL), for finite $N$ their accuracy suffers from additional stochastic error. In future work, we would like to understand why the polynomial estimator is more robust with respect to the decreased dimensionality and whether it is possible to modify the fast algebraic estimator accordingly to increase its robustness.

Lastly, the theoretical identification of the right root in the polynomial method is not fully rigorous and general yet. From the numerical study with both OAMP-based and AMP-based algorithms equipped with different types of denoisers we have found that the derivative \eqref{eq:quadratic_equation_derivative} was always positive. Similarly, we have observed that the scalar $u_3$ from the quadratic polynomial \eqref{eq:quadratic_equation_1} was always negative, which implies that $\alpha_t$ always corresponds to the smallest root within $\halfa_1$ and $\halfa_2$. We have observed these results (not demonstrated here) to hold for different condition numbers and different types of singular value distribution of $\bA$. The current proof technique in Section \ref{sec:tendency_in_root_identification} is developed specifically for a few examples of OAMP-based algorithms, while it seems that there is a more general result explaining the observed tendency for more broader range of SMP algorithms. We plan to investigate this observation in further works.

\section*{Appendix A}

\subsection*{Preliminary Lemmas}

Before proceeding to the proof of Theorem \ref{th:two_quadratic_equations}, we introduce several important lemmas that will be used further.

\begin{lemma} \label{lemma:normal_vectors_matrix_interaction}
    
    \cite{rand_mat_silverstein}, \cite{NS-VAMP}, \cite{AMP_SE_non_separable}: Let $\bx_1$ and $\bx_2$ be two $N$-sized zero-mean isotropic Gaussian vectors with variances $v_1$ and $v_2$ respectively and $\bD$ be a $N$ by $N$ symmetric positive semidefinite matrix independent of $\bx_1$ and $\bx_2$. Additionally assume that as $N \rightarrow \infty$, the Empirical Eigenvalue Distribution of $\bD$ converges to a density with compact support. Then almost surely we have that
    
    \noindent
    \begin{equation}
        \lim_{N \rightarrow \infty} \inv{N} \bx_1^T \bD \bx_2 \overset{a.s.}{=} \frac{\bx_1^T \bx_2}{N} \inv{N} Tr \big\{ \bD \big\} \label{eq:normal_vect_matrix_inner_prod}
    \end{equation}
    
\end{lemma}

\noindent
The following lemma defines certain Large System Limit (LSL) properties of the linear block in OAMP-based algorithms with the 1D SE.

\begin{lemma} \label{lemma:q_t_behaviour}
    
    \cite{VAMP}, \cite{UnifiedSE}: Let $\bq_{\tau} = \bs_{\tau} - \bx$ for $\tau \leq t$ be the error vector in OAMP-based algorithms with the 1D SE, $\bA$ be the measurement operator from \eqref{eq:y_measurements} and $\bD$ be as above in Lemma \ref{lemma:normal_vectors_matrix_interaction}. Then, under Assumptions 1-3 and assuming \eqref{eq:h_t_gaussian} - \eqref{eq:w_Aq_independence} hold up to iterations $t$, we have the following asymptotic results for $\tau, k \leq t$
    
    \noindent
    \begin{equation*}
        \lim_{N \rightarrow \infty} \inv{N} \bq_{\tau} \bA^T \bD (\bw - \bA \bq_{k}) \overset{a.s.}{=} - \frac{\bq_{\tau}^T \bq_k}{N} \inv{N} Tr \big\{ \bA \bA^T \bD \big\}
    \end{equation*}
    
    \noindent
    and
    
    \noindent
    \begin{equation}
        \lim_{N \rightarrow \infty} \inv{N} \bq_{\tau}^T \big( \bq_k + \gamma_t^{-1} \bA^T \bD (\bw - \bA \bq_k \big) \overset{a.s.}{=} 0
    \end{equation}
    
    \noindent
    where $\gamma_t = \inv{N} Tr \big\{ \bA \bA^T \bD \big\}$.
    
\end{lemma}

\subsection*{Proof of Theorem \ref{th:two_quadratic_equations}}

\begin{proof}[\unskip\nopunct]

Now we are ready to prove Theorem \ref{th:two_quadratic_equations}. We continue the idea introduced in Section \ref{sec:poly_div_est}, where we defined the parametrized denoiser $\bhs_{t+1} = \bhf(\br, \halfa)$ in \eqref{eq:bhf} and its oracle error $\bhq_{t+1}$ in \eqref{eq:bhq}. Additionally, we introduced two functions 

\noindent
\begin{gather}
    J_1(\halfa) = \inv{N} ||\bhs_{t+1}||^2 - 2 \inv{N} \br_t^T \bhs_{t+1} + v_x\\
    J_2(\halfa) = \inv{N}||\by - \bA \bhs_{t+1}||^2 - \delta v_w
\end{gather}

\noindent
where $v_x \overset{a.s.}{=} \lim_{N \rightarrow \infty} \inv{N} ||\bx||^2$ is the variance of the ground-truth signal $\bx$. Here we did and in the following we will drop the dependence of $\bhs_{t+1}$ and of $\bhq_{t+1}$ on the scalar parameters $\halfa$ to simplify the notations. 

The following lemma establishes the LSL properties of $J_1(\halfa)$ and of $J_2(\halfa)$

\begin{lemma} \label{lemma:E1_E2_LSL_model}

    Let the parametrized update $\bhs_{t+1} = \bhf(\br_t, \halfa)$ and its error $\bhq_{t+1}$ be as in \eqref{eq:bhf} and \eqref{eq:bhq} respectively. Define the following matrices
    
    \noindent
    \begin{gather*}
        \bH_{t+1} = \big( \bh_0, \bh_1,..., \bh_{t} \big) \\
        \bQ_{t+1} = \big( \bq_0, \bq_1,..., \bq_{t} \big) \\
        \bM_{t+1} = \bV^T \bH_{t+1} \\
        \bB_{t+1} = \bV^T \bQ_{t+1}
    \end{gather*}
    
    \noindent
    and vectors $\bhnu_{t+1} = \bQ_{t+1}^{\dagger} \bhq_{t+1}$ and $\bhphi_{t+1} = \bH_{t+1}^{\dagger} \bhq_{t+1}$, where $\dagger$ represents the pseudo-inverse operator. Lastly, define $\alpha_t$ to be the divergence of the denoiser $\bg_D(\br_t)$. Then, under Assumptions 1-3 and assuming \eqref{eq:h_t_gaussian} - \eqref{eq:w_Aq_independence} hold up to iterations $t$, the functions $J_1(\halfa)$ and $J_2(\halfa)$ are asymptotically equivalent to
    
    \noindent
    \begin{gather*}
        \lim_{N \rightarrow \infty} J_1(\halfa) \overset{a.s.}{=} \lim_{N \rightarrow \infty} \inv{N} || \bhs_{t+1} - \bx ||^2 + \eta_{t+1}(\halfa)  \\
        \lim_{N \rightarrow \infty} J_2(\halfa) \overset{a.s.}{=} \lim_{N \rightarrow \infty} \inv{N} || \bhs_{t+1} - \bx ||^2 + \beta_{t+1}(\halfa)
    \end{gather*}
    
    \noindent
    where
    
    \noindent
    \begin{gather*}
        \eta_{t+1}(\halfa) = 2 (\alpha_t - \halfa) v_{h_t} \\
        \begin{split}
            &\beta_{t+1}(\halfa) = \inv{N} \Big( ||\bS \bM_{t+1} \bhphi_{t+1}||^2 - ||\bM_{t+1} \bhphi_{t+1} ||^2 \Big) \\
            &+ 2 \inv{N} \Big( \bhnu_{t+1}^T \bB_{t+1}^T \bS^T \bS \bM_{t+1} \bhphi_{t+1} - \bw^T \bA \bhq_{t+1} \Big)
        \end{split}
    \end{gather*}

    \noindent
    Additionally, when $\halfa = \alpha_t$ we have that
    
    \noindent
    \begin{gather*}
        \lim_{N \rightarrow \infty} \eta_{t+1}(\alpha_t) \overset{a.s.}{=} 0  \\
        \lim_{N \rightarrow \infty} \beta_{t+1}(\alpha_t) \overset{a.s.}{=} 0 
    \end{gather*}

\end{lemma}

\begin{proof}
    For the rest of the proof, we let the Assumptions 1-3 and the asymptotic identities \eqref{eq:h_t_gaussian} - \eqref{eq:w_Aq_independence} hold up to iteration $t$. Then, we start with the following norm

\noindent
\begin{align*}
    \lim_{N \rightarrow \infty} &\inv{N} ||\bhs_{t+1} - \bx||^2 \overset{a.s.}{=} \lim_{N \rightarrow \infty} \inv{N} ||\bhs_{t+1}||^2 + v_x - \frac{2}{N} \bx^T \bhs_{t+1} \nonumber\\
    &= \lim_{N \rightarrow \infty} \inv{N} ||\bhs_{t+1}||^2 + v_x - \frac{2}{N} \big( \br_t^T \bhs_{t+1} - \bh_t^T \bhs_{t+1} \big) 
\end{align*}

\noindent
where we used the identity $\bx = \br_t - \bh_t$. Next, from the definition of $\bhs_{t+1}$ we have that

\noindent
\begin{align}
    \lim_{N \rightarrow \infty} \inv{N} \bh_t^T \bhs_{t+1}& = \lim_{N \rightarrow \infty} \inv{N} \bh_t^T (\bg_D(\br_t) - \halfa \br_t) \nonumber\\
    &\overset{a.s.}{=} \alpha_t v_{h_t} - \halfa_t v_{h_t}
\end{align}

\noindent
where we used the Strong Law of Large Numbers and Stein's Lemma \cite{SURE}. Finally, by defining $\eta_{t+1}(\halfa) = 2 (\alpha_t - \halfa) v_{h_t}$ and $J_1(\halfa) = \inv{N} ||\bhs_{t+1}||^2 - 2 \br_t^T \bhs_{t+1} + v_x$ we get the asymptotic identity for $J_1(\halfa)$. 

Next, to prove the identity for $J_2(\halfa)$, we can use \eqref{eq:y_measurements} to show that

\noindent
\begin{align}
    &\inv{N} ||\by - \bA \bhs_{t+1}||^2 = \inv{N} ||\bw - \bA \bhq_{t+1}||^2 \nonumber\\
    &= \inv{N} ||\bw||^2 + \inv{N} ||\bA \bhq_{t+1}||^2 - 2 \inv{N} \bw^T \bA \bhq_{t+1}
\end{align}

\noindent
Define the vector $\bhb_{t+1} = \bV^T \bhq_{t+1}$. Then, under the above assumptions, it was shown in \cite{UnifiedSE} that $\bhb_{t+1}$ has the following evolution model

\noindent
\begin{equation}
    \bhb_{t+1} = \bb_{t+1}(\halfa) = \bB_{t+1} \bhnu_{t+1} + \bM_{t+1} \bhphi_{t+1} + \bhu_{t+1}
\end{equation}

\noindent
where the matrices $\bB_{t+1}$ and $\bM_{t+1}$ and vectors $\bhnu_{t+1}$ and $\bhphi_{t+1}$ are as in Lemma \ref{lemma:E1_E2_LSL_model} and $\bhu_t$ is a zero-mean i.i.d. Gaussian vector independent of $\bB_{t+1}$ and of $\bM_{t+1}$. Additionally, one can show that the matrices $\bB_{t+1}$ and $\bM_{t+1}$ are orthogonal to each other \cite{UnifiedSE}. This implies that the variance of $\bhq_{t+1}$ can be written as

\noindent
\begin{align*}
    &\lim_{N \rightarrow \infty} \inv{N} ||\bhq_{t+1}||^2 =  \lim_{N \rightarrow \infty} \inv{N} ||\bhb_{t+1}||^2 \nonumber\\
    &\overset{a.s.}{=} \lim_{N \rightarrow \infty} \inv{N} \Big( ||\bB_{t+1} \bhnu_{t+1}||^2 +  ||\bM_{t+1} \bhphi_{t+1}||^2 + ||\bhu_{t+1}||^2 \Big)
\end{align*}

\noindent
Next, we can use the SVD of $\bA$, the fact that $\bB_{t+1}$ is a set of zero-mean Gaussian vectors and Lemma \ref{lemma:normal_vectors_matrix_interaction} to show that

\noindent
\begin{equation}
    \lim_{N \rightarrow \infty} \inv{N} ||\bS \bB_{t+1} \bhnu_{t+1}||^2 \overset{a.s.}{=} \lim_{N \rightarrow \infty} \inv{N} ||\bB_{t+1} \bhnu_{t+1}||^2
\end{equation}

\noindent
With these results we can obtain

\noindent
\begin{align}
    &\lim_{N \rightarrow \infty} \inv{N} ||\bA \bhq_{t+1}||^2 = \lim_{N \rightarrow \infty} \inv{N} || \bS \bhb_{t+1}||^2 \nonumber\\
    &= \lim_{N \rightarrow \infty} \inv{N} ||\bS \bB_{t+1} \bhnu_{t+1}||^2 + \inv{N} ||\bS \bM_{t+1} \bhphi_{t+1}||^2 \nonumber\\
    &+ \inv{N} ||\bS \bhu_{t+1}||^2 + 2 \inv{N} \bhnu_{t+1}^T \bB_{t+1}^T \bS^T \bS \bM_{t+1} \bhphi_{t+1} \nonumber\\
    &\overset{a.s.}{=} \lim_{N \rightarrow \infty} \inv{N} ||\bB_{t+1} \bhnu_{t+1}||^2 + \inv{N} ||\bS \bM_{t+1} \bhphi_{t+1}||^2 \nonumber\\
    &+ \inv{N} ||\bhu_{t+1}||^2 + 2 \inv{N} \bhnu_{t+1}^T \bB_{t+1}^T \bS^T \bS \bM_{t+1} \bhphi_{t+1}
\end{align}

\noindent
By adding and subtracting the norm $\inv{N} || \bM_{t+1} \bhphi_{t+1}||^2$ and defining

\noindent
\begin{align*}
    &\zeta_{t+1}(\halfa) = \inv{N} ||\bS \bM_{t+1} \bhphi_{t+1}||^2 - \inv{N} ||\bM_{t+1} \bhphi_{t+1}||^2 \nonumber\\
    &+ 2 \inv{N} \bhnu_{t+1}^T \bB_{t+1}^T \bS^T \bS \bM_{t+1} \bhphi_{t+1}
\end{align*}

\noindent
we can show that

\noindent
\begin{equation}
    \lim_{N \rightarrow \infty} \inv{N} ||\bA \bhq_{t+1}||^2 \overset{a.s.}{=} \inv{N} ||\bhq_{t+1}||^2 + \zeta_{t+1}(\halfa)
\end{equation}

\noindent
Then, if we set $\beta_{t+1}(\halfa) = \zeta_{t+1}(\halfa) - 2 \inv{N} \bw^T \bA \bhq_{t+1}$, we obtain the desired LSL result for $J_2(\halfa)$. 

Lastly, we can use \eqref{eq:h_q_independence} and \eqref{eq:w_Aq_independence} to show that for $\halfa = \alpha_t$ we have $\lim_{N \rightarrow \infty} \bhphi_{t+1} \overset{a.s.}{=} 0$ and $\lim_{N \rightarrow \infty} \inv{N} \bw^T \bA \bhq_{t+1} \overset{a.s.}{=} 0$, which implies that $\lim_{N \rightarrow \infty} \beta_{t+1}(\halfa) \overset{a.s.}{=} 0$. \end{proof}

\noindent
Lemma \ref{lemma:E1_E2_LSL_model} suggests that $\alpha_t$ is one of the roots to the following equation

\noindent
\begin{equation*}
    \lim_{N \rightarrow \infty} E^{[1]}(\halfa) = \lim_{N \rightarrow \infty} J_1(\halfa) - J_2(\halfa) = 0 
\end{equation*}

\noindent
Next we show that this equation is equivalent to \eqref{eq:quadratic_equation_1}. To increase the readability, in the following we use $\bg_D$ to refer to $\bg_D(\br_t)$. First, we expand the norm in $J_2(\halfa)$ to obtain

\noindent
\begin{equation}
    ||\by - \bA \bhs_{t+1}||^2 = ||\by||^2 + ||\bA \bhs_{t+1}||^2 - 2 \by^T \bA \bhs_{t+1} \label{eq:norm_y_minus_A_s_hat}
\end{equation}

\noindent
Using the definition of $\bhs_{t+1}$, we can expand the norm $||\bA \bhs_{t+1}||^2$ as

\noindent
\begin{equation}
    ||\bA \bhs_{t+1}||^2 = c_1 - c_2 \halfa + c_3 \halfa^2 \label{eq:norm_A_s_hat}
\end{equation}

\noindent
where the scalars $c_1$, $c_2$ and $c_3$ are

\noindent
\begin{equation*}
    c_1 = ||\bA \bg_D||^2 \quad c_2 = 2 \bg_D^T \bA^T \bA \br_t \quad c_3 = ||\bA \br_t||^2
\end{equation*}

\noindent
Similarly we expand the inner-product $\by^T \bA \bhs_{t+1}$

\noindent
\begin{equation}
    \inv{N} \by^T \bA \bhs_{t+1} = d_1 - d_2 \halfa \label{eq:y_A_s_hat}
\end{equation}

\noindent
with

\noindent
\begin{equation*}
    d_1 = \by^T \bA \bg_D \quad d_2 = \by^T \bA \br_t
\end{equation*}

\noindent
By defining a scalar $c_y = ||\by||^2$ and using \eqref{eq:norm_A_s_hat} and \eqref{eq:y_A_s_hat} in \eqref{eq:norm_y_minus_A_s_hat}, we obtain

\noindent
\begin{align}
    &||\by - \bA \bhs_{t+1}||^2 = c_y + c_1 - c_2 \halfa + c_3 \halfa^2 - 2 \big( d_1 - d_2 \halfa \big) \nonumber\\
    &= \big( c_y + c_1 - 2 d_1 \big) + \big( 2 d_2 - c_2 \big) \halfa + c_3 \halfa^2 \label{eq:norm_y_minus_A_s_hat_simplified}
\end{align}

\noindent
In the similar way we can expand the norm in $J_1(\halfa)$. Since

\noindent
\begin{gather}
    ||\bhs_{t+1}||^2 = ||\bg_D||^2 - 2 \halfa \br_t^T \bg_D + \halfa^2 ||\br_t||^2 \\
    \br_t^T \bhs_{t+1} = \br_t^T \bg_D - \halfa ||\br_t||^2
\end{gather}

\noindent
we can show that $J_1(\halfa)$ is equal to

\noindent
\begin{align}
    J_1(\halfa) = \inv{N} \big( m_1 + m_2 \halfa + m_3 \halfa^2 \big) \label{eq:norm_s_hat_minus_r}
\end{align}

\noindent
where we used the fact that $v_x \overset{a.s.}{=} \lim_{N \rightarrow \infty} \inv{N} ||\br_t||^2 - v_{h_t}$ and defined

\noindent
\begin{equation*}
    m_1 = ||\bg_D - \br_t||^2 \quad m_2 = 2 \br_t^T ( \br_t - \bg_D ) \quad m_3 = || \br_t||^2
\end{equation*}

\noindent
Using \eqref{eq:norm_y_minus_A_s_hat_simplified} and \eqref{eq:norm_s_hat_minus_r}, we can show that the difference $J_1(\halfa) - J_2(\halfa)$ is equal to

\noindent
\begin{align}
    &J_1(\halfa) - J_2(\halfa) = \inv{N} \Big( m_1 + m_2 \halfa + m_3 \halfa^2 - v_{h_t} \nonumber\\
    &- \big( c_y + c_1 - 2 d_1 \big) - \big( 2 d_2 - c_2 \big) \halfa - c_3 \halfa^2 + \delta v_w \Big) \nonumber\\
    &= u_1 + u_2 \halfa + u_3 \halfa^2
\end{align}

\noindent
where $u_1$ is

\noindent
\begin{align*}
    &u_1 = \inv{N} ( m_1 - v_{h_t} - c_y - c_1 + 2d_1 + \delta v_w) \nonumber\\
    &= \inv{N} ||\bg_D - \br_t||^2 - v_{h_t} - \big( \inv{N} ||\bA \bg_D - \by||^2 - \delta v_w \big)
\end{align*}

\noindent
$u_2$ is

\noindent
\begin{align*}
    &u_2 = \inv{N} \big( m_2 - 2 d_2 + c_2 \big) \nonumber\\
    &= \frac{2}{N} \Big( ( \br_t - \bg_D )^T \br_t - (\by - \bA \bg_D)^T \bA \br_t \Big)
\end{align*}

\noindent
and $u_3$ is

\noindent
\begin{equation*}
    u_3 = \inv{N} (m_3 - c_3) = \inv{N} ||\br_t||^2 - \inv{N} ||\bA \br_t||^2
\end{equation*}

\noindent
which completes the proof. \end{proof}

\section*{Appendix B}
\begin{proof}[\unskip\nopunct]

The proof of Theorem \ref{th:divergence_identity_OAMP_based_algorithms} for MF-OAMP, VAMP and CG-VAMP is based on the asymptotic result formulated in the following theorem.

\begin{theorem} \label{th:divergence_general_identity}
    
    Consider an OAMP-based algorithm \eqref{eq:LB}-\eqref{eq:DB}, where the denoising step uses a denoiser $\bg_D(\br_t)$, $\br_t = \bx + \bh_t$ and let $\alpha_t$ be the divergence of $\bg_D$ at $\br_t$. Let $\bff_L$ have a form
    
    \noindent
    \begin{equation}
        \bff_L(\bZ_t) = \inv{\sum_{\tau = 0}^t \gamma_t^{\tau}} \Big( \sum_{\tau = 0}^t \gamma_t^{\tau} \bs_{\tau} + \bA^T \bF_t(\bZ_t) \Big) \label{eq:th4_f_L}
    \end{equation}
    
    \noindent
    where $\bZ_t = \big(\bz_t, \bz_{t-1},..., \bz_0 \big)$, $\bz_{\tau} = \by - \bA \bs_{\tau}$ and 
    
    \noindent
    \begin{equation}
        \gamma_{\tau}^t = \inv{N} Tr \Big\{ \bA \bJ_{\bs_{\tau}} \big( \bF_t(\bZ_t) \big) \Big\} \label{eq:th4_gamma}
    \end{equation}
    
    \noindent
    Additionally, let $\bF_t$ have a finite Jacobian $\bJ_{\bs_{\tau}} \big( \bF_t(\bZ_t) \big)$ with respect to each $\bs_{\tau}$, $\tau \leq t$ as $N \rightarrow \infty$. Define a corrected denoiser
    
    \noindent
    \begin{equation}
        \bbf(\br_t) = \bg_D(\br_t) - \alpha_t \br_t \label{eq:bar_f}
    \end{equation}
    
    \noindent
    and its error
    
    \noindent
    \begin{equation*}
        \bbq_{t+1} = \bbf(\br_t) - \bx \label{eq:bar_q}
    \end{equation*}
    
    \noindent
    Then, under Assumptions 1-3 and assuming \eqref{eq:h_t_gaussian} - \eqref{eq:w_Aq_independence} hold up to iteration $t$, the derivative \eqref{eq:quadratic_equation_derivative} at $\halfa = \alpha_t$ almost surely converges to
    
    \noindent
    \begin{equation}
        \lim_{N \rightarrow \infty} u_2 + 2 u_3 \alpha_t \overset{a.s.}{=} \lim_{N \rightarrow \infty} 2 \big( v_{h_t} + \inv{N} \bh_t^T \bA^T \bA \bbq_{t+1} - v_w \big) \label{eq:th_LSL_poly_deriv}
    \end{equation}
    
\end{theorem}

\begin{proof}
    In the following, let Assumptions 1-3 and the asymptotic identities \eqref{eq:h_t_gaussian} - \eqref{eq:w_Aq_independence} hold up to iteration $t$ and $\alpha_t$ be the divergence of $\bg_D(\br_t)$. Then, we can use the definitions of $u_2$ and $u_3$ from Theorem \ref{th:two_quadratic_equations} to show that the derivative \eqref{eq:quadratic_equation_derivative} at $\halfa = \alpha_t$ is equal to

    \noindent
    \begin{align}
        &u_2 + 2 u_3 \alpha = \frac{2}{N} \bigg(  ||\br_t||^2 - \br_t^T \big( \bg_D(\br_t) - \alpha \br_t \big) - \br_t^T \bA^T \by \nonumber\\
        &+ \br_t^T \bA^T \bA \big( \bg_D(\br_t) - \alpha \br_t \big) \bigg) \label{eq:poly_deriv_full}
    \end{align}
    
    \noindent
    Note that by defining the output of the corrected denoiser $\bbs_{t+1} = \bbf(\br_t)$, we can rewrite \eqref{eq:poly_deriv_full} as
    
    \noindent
    \begin{align}
         & u_2 + 2 u_3 \alpha = \frac{2}{N} \big(  ||\br_t||^2 - \br_t^T \bbs_{t+1} - \br_t^T \bA^T \big( \by - \bA \bbs_{t+1} \big) \big) \nonumber\\
         &= \frac{2}{N} \bigg( ||\br_t||^2 - \br_t^T \Big( \bbs_{t+1} + \bA^T \big( \by - \bA \bbs_{t+1} \big) \Big) \bigg) \label{eq:deriv_pre_final_result}
    \end{align}
    
    \noindent
    Let $\br_t = \bx + \bh_t$ be updated as $\br_t = \bff_L(\bZ_t)$ with $\bff_L$ defined in \eqref{eq:th4_f_L}. Then, using the SVD of $\bA = \bU \bS \bV^T$, the definition $\bb_{\tau} = \bV^T \bq_{\tau}$ and \eqref{eq:w_Aq_independence}, we can show that
    
    \noindent
    \begin{equation}
        \lim_{N \rightarrow \infty} \inv{N} \bw^T \bA \bh_t \overset{a.s.}{=} \lim_{N \rightarrow \infty} \inv{\sum_{\tau = 0}^t \gamma_t^{\tau}} \inv{N} \bw^T \bA \bA^T \bF_t (\bZ_t) \label{eq:w_A_h}
    \end{equation}
    
    \noindent
    where $\bz_{\tau} = \bw - \bA \bq_{\tau}$. In the following, let $\bF_t$ have a finite Jacobian with respect to each $\bs_{\tau}$, $\tau \leq t$ as $N \rightarrow \infty$. Then, by using the Law of Large Numbers, Stein's Lemma \cite{SURE}, \eqref{eq:b_t_gaussian} and \eqref{eq:w_Aq_independence}, we obtain
    
    \noindent
    \begin{align}
        &\lim_{N \rightarrow \infty} \inv{N} \bw^T \bA \bA^T \bF_t (\bZ_t) \overset{a.s.}{=} \lim_{N \rightarrow \infty} \expectation \bigg[ \inv{N} \bw^T \bA \bA^T \bF_t (\bZ_t) \bigg] \nonumber\\
        &= \lim_{N \rightarrow \infty} \inv{N} \expectation \bigg[ \inv{M} Tr \Big\{ \bw (\bF_t (\bZ_t))^T \bA \bA^T \Big\} \bigg] \nonumber\\
        &\overset{a.s.}{=} \lim_{N \rightarrow \infty} \inv{N} \inv{M} Tr \bigg\{ \bA \bA^T \sum_{\tau = 0}^t \expectation \big[ \bw \bz_{\tau}^T \big] \bJ_{\bz_{\tau}} \big( \bF_t(\bZ_t) \big) \bigg\} \nonumber\\
        &\overset{a.s.}{=} \lim_{N \rightarrow \infty} \inv{N} \inv{M} Tr \bigg\{ \bA \bA^T \sum_{\tau = 0}^t \expectation \big[ \bw \bw^T \big] \bJ_{\bz_{\tau}} \big( \bF_t(\bZ_t) \big) \bigg\} \nonumber\\
        &\overset{a.s.}{=} \lim_{N \rightarrow \infty} v_w \sum_{\tau = 0}^t \inv{N} Tr \Big\{ \bA \bA^T \bJ_{\bz_{\tau}} \big( \bF_t(\bZ_t) \big) \Big\} = v_w \sum_{\tau = 0}^t \gamma_{\tau}^t \label{eq:w_A_At_F}
    \end{align}
    
    \noindent
    where the last step comes from the definition of $\bz_{\tau}$ and of $\gamma_{\tau}^t$. Therefore we conclude that \eqref{eq:w_A_h} almost surely converges to $v_w$ and, together with \eqref{eq:h_q_independence}, this implies that 
    
    \noindent
    \begin{align}
        \lim_{N \rightarrow \infty} \inv{N} \bh_t^T &\Big( \bx + \bbq_{t+1} + \bA^T \big( \bw - \bA \bbq_{t+1} \big) \Big) \nonumber\\
        &\overset{a.s.}{=} v_w - \inv{N} \bh_t^T \bA^T \bA \bbq_{t+1} \label{eq:h_h_hat}
    \end{align}
    
    \noindent
     Additionally, since $\by - \bA \bbs_{t+1} = \bw - \bA \bbq_{t+1}$, $\br_t = \bx + \bh_t$ and $\bx = -\bq_0$, we can use the Strong Law of Large Numbers and Lemma \ref{lemma:q_t_behaviour} to obtain
    
    \noindent
    \begin{align*}
        &\lim_{N \rightarrow \infty} \inv{N} \bx^T \Big( \bx + \bbq_{t+1} + \bA^T \big( \by - \bA \bbs_{t+1} \big) \Big) \nonumber\\
        &\overset{a.s.}{=} v_x - \lim_{N \rightarrow \infty} \inv{N} \bq_0^T \Big( \bbq_{t+1} + \bA^T (\bw - \bA \bbq_{t+1}) \overset{a.s.}{=} v_x
    \end{align*}
    
    \noindent
    Together with $\eqref{eq:h_h_hat}$, \eqref{eq:deriv_pre_final_result} and the fact that $\lim_{N \rightarrow \infty} \inv{N} ||\br_t||^2 \overset{a.s.}{=} v_x + v_{h_t}$, this implies the result \eqref{eq:th_LSL_poly_deriv}. \end{proof}

To proceed next, first, we mention that for MF-OAMP, VAMP and CG-VAMP we have $\bF_t(\bZ_t) = \bF_t(\bz_t)$. In those cases, we can do the same steps as in \eqref{eq:w_A_At_F} to show that

\noindent
\begin{equation}
    \gamma_t^t = \gamma_t = \lim_{N \rightarrow \infty} -\frac{\bq_t^T \bA^T \bF_t(\bz_t)}{N v_{q_t}} \label{eq:gamma_general}
\end{equation}

\noindent
and $\gamma_{\tau}^t = 0$ for $\tau < t$. In particular, substituting $\bF_t(\bz_t) = \bz_t$ into \eqref{eq:th4_f_L} for MF-OAMP, $\bF_t(\bz_t) = \bW_t^{-1} \bz_t$ for VAMP, leads us to the correction scalars $\gamma_t = 1$ for MF-OAMP and as in \eqref{eq:gamma_VAMP} for VAMP. For CG-VAMP we can use Lemma 1 from \cite{CG_EP} to show that as $N \rightarrow \infty$ the CG algorithm approximating the SLE \eqref{eq:SLE} almost surely converges to a matrix polynomial of $\bA \bA^T$, which implies $J_{s_t} \bF_t(\bz_t)$ is finite. Thus, we have a well-defined $\gamma_t^t = \gamma_t$ and $\gamma_{\tau} = 0$, $\tau < t$ for CG-VAMP. Lastly, the linear update $\bff_L$ with $\bF_t(\bZ_t)$ in WS-CG-VAMP also fits into the model \eqref{eq:th4_f_L} as shown in Theorem 3 in \cite{OurPaper}. Thus, the result \eqref{eq:th_LSL_poly_deriv} holds for MF-OAMP, VAMP, CG-VAMP and WS-CG-VAMP.

Next we finish the proof of Theorem \ref{th:divergence_identity_OAMP_based_algorithms} for MF-OAMP and VAMP algorithms, while the proof for CG-VAMP is presented in the supplementary materials. We begin with noting that

\noindent
\begin{align}
    &\lim_{N \rightarrow \infty} \inv{N} \bh_t^T \bA^T \bA \bbq_{t+1} \nonumber\\
    &= \lim_{N \rightarrow \infty} \inv{N} \Big( \bq_t + \gamma_t^{-1} \bA^T \bF_t(\bz_t) \Big)^T \bA^T \bA \bbq_{t+1} \nonumber\\
    &= \lim_{N \rightarrow \infty} \inv{N} \bq_t^T \bA^T \bA \bbq_{t+1} + \gamma_t^{-1} \inv{N} (\bF_t(\bz_t))^T\bA \bA^T \bA \bbq_{t+1} \nonumber\\
    &\overset{a.s.}{=} \lim_{N \rightarrow \infty} \psi_t + \gamma_t^{-1} \inv{N} (\bF_t(\bz_t))^T\bA \bA^T \bA \bbq_{t+1} \label{eq:OAMP_h_At_A_q_hat}
\end{align}

\noindent
where we used $\inv{N} Tr \big\{ \bA \bA^T \big\} = 1$ and defined a scalar $\psi_t = \inv{N} \bq_t^T \bbq_{t+1}$. Next, we consider MF-OAMP and VAMP separately.

\subsection{MF-OAMP}

First, using the definition of $\bff_L$ in MF-OAMP as above and Lemma \ref{lemma:q_t_behaviour}, we can obtain

\noindent
\begin{align}
    \lim_{N \rightarrow \infty} \inv{N} &(\bF_t(\bz_t))^T\bA \bA^T \bA \bbq_{t+1} = \lim_{N \rightarrow \infty} \bz_t^T \bA \bA^T \bA \bbq_{t+1} \nonumber\\
    &\overset{a.s.}{=} \lim_{N \rightarrow \infty} - \psi_t \inv{N} Tr \big\{ \Lambda^2 \big\} = - \psi_t \chi_2 \label{eq:Ft_A_At_A_q_bar}
\end{align}

\noindent
where $\chi_j$ is as in \eqref{eq:chi}. Similarly, we can use \eqref{eq:gamma_general} to show that 

\noindent
\begin{align}
    &v_{h_t} \overset{a.s.}{=} \lim_{N \rightarrow \infty} \inv{N} ||\bh_t||^2 =  \lim_{N \rightarrow \infty} \inv{N} || \bq_t + \bA^T \bz_t ||^2 \nonumber\\
    &\overset{a.s.}{=} \lim_{N \rightarrow \infty} v_{q_t} + 2 \inv{N} \bq_t^T \bA^T \bz_t + \inv{N} || \bA^T \bz_t ||^2 \nonumber\\
    &\overset{a.s.}{=} v_w + v_{q_t} \chi_2 - v_{q_t} \label{eq:v_h}
\end{align}

\noindent
where we also used \eqref{eq:w_Aq_independence}. Substituting \eqref{eq:OAMP_h_At_A_q_hat}, \eqref{eq:Ft_A_At_A_q_bar} and \eqref{eq:v_h} into \eqref{eq:th_LSL_poly_deriv} leads to the final result for MF-OAMP

\noindent
\begin{align}
    &\lim_{N \rightarrow \infty} \frac{1}{2} \big( u_2 + 2 \alpha u_3 \big) \overset{a.s.}{=} ( \chi_2 - 1) v_{q_t} + \psi_t - \psi_t \chi_2 \nonumber\\
    &= \big( \psi_t - v_{q_t} \big) - \chi_2 \big( \psi_t - v_{q_t} \big) = \big( 1 - \chi_2 \big) \big( \psi_t - v_{q_t} \big) \nonumber
\end{align}

\subsection{VAMP}

For VAMP, we begin by noticing that

\noindent
\begin{equation}
    \bA \bA^T = \frac{\bW_t - v_w \bI}{v_{q_t}}
\end{equation}

\noindent
which implies

\noindent
\begin{align}
    &\lim_{N \rightarrow \infty} \inv{N} (\bF_t(\bz_t))^T\bA \bA^T \bA \bbq_{t+1} \nonumber\\
    &= \lim_{N \rightarrow \infty} \frac{\inv{N} (\bF_t(\bz_t))^T\bW_t  \bA \bbq_{t+1} - v_w \inv{N} (\bF_t(\bz_t))^T\bA \bbq_{t+1}}{v_{q_t}} \nonumber\\
    &\overset{a.s.}{=} \lim_{N \rightarrow \infty} \frac{\inv{N}  (\bF_t(\bz_t))^T\bW_t  \bA \bbq_{t+1} + v_w \gamma_t \psi_t }{v_{q_t}} \label{eq:OAMP_based_alg_second_comp_general}
\end{align}

\noindent
where we also used \eqref{eq:gamma_general}. Since in VAMP $\bF_t(\bz_t) = \bW_t^{-1} \bz_t$, we use \eqref{eq:w_Aq_independence} to obtain

\noindent
\begin{align*}
    &\lim_{N \rightarrow \infty} \inv{N} (\bF_t(\bz_t))^T\bA \bA^T \bA \bbq_{t+1} \nonumber\\
    &= \lim_{N \rightarrow \infty} \frac{\inv{N} \bz_t^T \bW_t^{-1} \bW_t  \bA \bbq_{t+1} + v_w \gamma_t \psi_t }{v_{q_t}} \overset{a.s.}{=} \frac{-\psi_t  + v_w \gamma_t \psi_t }{v_{q_t}} 
\end{align*}

\noindent
Using this result together with \eqref{eq:OAMP_h_At_A_q_hat} and the fact that for VAMP the variance $v_{h_t}$ is updated as \cite{EP_Keigo}

\noindent
\begin{equation}
    v_{h_t} = \gamma_t^{-1} - v_{q_t}, \label{eq:VAMP_v_h}
\end{equation}

\noindent
we can obtain

\noindent
\begin{align}
    &\lim_{N \rightarrow \infty} \frac{1}{2} \big( u_2 + 2 \alpha u_3 \big) \nonumber\\
    &\overset{a.s.}{=} \psi_t - v_{q_t} + \gamma_t^{-1} - \gamma_t^{-1} \frac{\psi_t}{v_{q_t}} + \frac{v_w \psi_t}{v_{q_t}} - v_w \nonumber\\
    &= \psi_t - v_{q_t} - \gamma_t^{-1} v_{q_t}^{-1} \Big( \psi_t - v_{q_t} \Big) + \frac{v_w \psi_t - v_{q_t} v_w }{v_{q_t}} \nonumber\\
    &= \big( 1 + (v_w - \gamma_t^{-1}) v_{q_t}^{-1} \big) \Big( \psi_t - v_{q_t} \Big) \nonumber\\
    &= \frac{\big(v_w + v_{q_t} - \gamma_t^{-1} \big)}{v_{q_t}} \Big( \psi_t - v_{q_t} \Big) = \frac{\big(v_w - v_{h_t} \big)}{v_{q_t}} \Big( \psi_t - v_{q_t} \Big) \nonumber
\end{align}

\noindent
which completes the proof for VAMP. \end{proof}

\section*{Appendix C}
\begin{proof}[\unskip\nopunct]

Here we prove Lemma \ref{lemma:v_h_minus_v_w_positivity}, which analyzes the sign of the difference $v_{h_t} - v_w$. From the update rule \eqref{eq:VAMP_v_h} of $v_{h_t}$ we have that

\noindent
\begin{equation}
    v_{h_t} - v_w = \gamma_t^{-1} - v_{q_t} - v_w = \frac{1 - \gamma_t(v_{q_t} + v_w)}{\gamma_t} \label{eq:v_h_minus_v_w}
\end{equation}

\noindent
Because $\gamma_t$ is strictly positive \cite{EP_Keigo}, for the positivity of \eqref{eq:v_h_minus_v_w} it is sufficient to have $1 - \gamma_t(v_{q_t} + v_w)$ to be positive. To proceed next, we notice that as $N \rightarrow \infty$, we can define $\gamma_t$ through the limiting eigenvalue distribution $p(\lambda)$ of $\Lambda = \bS \bS^T$ as \cite{EP_Keigo}

\noindent
\begin{equation}
    \lim_{N \rightarrow \infty} \gamma_t = \delta \int \frac{\lambda}{v_w + v_{q_t} \lambda} p(\lambda) d \lambda \label{eq:VAMP_gamma}
\end{equation}

\noindent
With \eqref{eq:VAMP_gamma}, we can rewrite the term $\gamma_t(v_{q_t} + v_w)$ as

\noindent
\begin{align}
    &\gamma_t(v_{q_t} + v_w) = \delta \int \frac{\lambda(v_{q_t} + v_w)}{v_w + v_{q_t} \lambda} p(\lambda) d \lambda \nonumber\\
    &= \delta \int \frac{v_{q_t} \lambda + v_w - v_w + v_w \lambda}{v_w + v_{q_t} \lambda} p(\lambda) d \lambda \nonumber\\
    &= \delta \Big( 1 + \int \frac{v_w \lambda - v_w}{v_w + v_{q_t} \lambda} p(\lambda) d \lambda \Big)
\end{align}

\noindent
Which implies that the nominator of \eqref{eq:v_h_minus_v_w} is equivalent to

\noindent
\begin{align}
    &1 - \gamma_t(v_{q_t} + v_w) = \delta \Big( \delta^{-1} - 1 - \int \frac{v_w \lambda - v_w}{v_w + v_{q_t} \lambda} p(\lambda) d \lambda \Big) \nonumber\\
    &= \delta \Big( \int \frac{ (\delta^{-1} - 1)(v_w + v_{q_t} \lambda) - v_w \lambda + v_w}{v_w + v_{q_t} \lambda} p(\lambda) d \lambda \Big) \nonumber\\
    &= \delta \Big( \int \frac{ \big( (\delta^{-1} - 1)v_{q_t} - v_w \big) \lambda + \delta^{-1} v_w}{v_w + v_{q_t} \lambda} p(\lambda) d \lambda \Big) \label{eq:lemma_one_proof_final_step}
\end{align}

\noindent
Since all the scalar variables in \eqref{eq:lemma_one_proof_final_step} are positive and $p(\lambda)$ is a proper probability density function, it is sufficient to have the following inequality
\noindent
\begin{equation}
    v_{q_t} > \frac{v_w}{\delta^{-1} - 1}
\end{equation}

\noindent
for \eqref{eq:lemma_one_proof_final_step} to be positive, which implies positivity of \eqref{eq:v_h_minus_v_w}. \end{proof}

\printbibliography

\end{document}

\maketitle

\begin{abstract}
    In these supplementary materials we consider two sub-problems of the main paper ``Divergence Estimation in Message Passing algorithms". In particular, we provide the proof of Theorem 3 and the motivation for Proposal 1 for the Conjugate Gradient (CG) Vector Approximate Message Passing (CG-VAMP) algorithm.
\end{abstract}

\section{Context of the problem}

To make the document self-contained, first, we set up the context of the problem. Here, we consider the CG-VAMP algorithm recovering $\bx \in \mathbb{R}^N$ from a set of linear measurements 

\noindent
\begin{equation}
    \by = \bA \bx + \bw
    \label{eq:y_measurements}
\end{equation}

\noindent
where $\by \in \mathbb{R}^M$ is the set of measurements, $\bw \in \mathbb{R}^M$  is a zero-mean i.i.d. Gaussian noise vector $\bw \sim \normDensity(0,v_w \bI_M)$ and $\bA \in \mathbb{R}^{M \times N}$. Following the notation from the main paper, the CG-VAMP algorithms consists of the linear step

\noindent
\begin{equation}
    \bff_L \big( \by, \bs_t \big) = \bs_t + \gamma_t^{-1} \bA^T \bF_t \big( \by - \bA \bs_t \big) \label{eq:f_l_VAMP}
\end{equation}

\noindent
and of the denoising step

\noindent
\begin{equation}
    \bff_D(\br_t) = C_t \Big( \bg_D(\br_t) - \alpha_t \br_t \Big) \label{eq:f_d_VAMP}
\end{equation}

\noindent
In \eqref{eq:f_d_VAMP}, $\bg_D$ represents a denoiser and $\alpha_t$ corresponds to its divergence at $\br_t$, while $C_t$ is some finite scalar. On the other hand, \eqref{eq:f_l_VAMP} involves $\bF_t \big( \by - \bA \bs_t \big)$ that is the output of the CG algorithm with $i$ iterations approximating the system of linear equations (SLE) 

\noindent
\begin{equation}
    \bW_t \bmu_t = \bz_t \label{eq:SLE}
\end{equation}

\noindent
for $\mu_t$, where

\noindent
\begin{equation}
    \bz_t = \by - \bA \bs_t 
\end{equation}

\noindent
and 

\noindent
\begin{equation}
    \bW_t = v_w \bI_M + v_{q_t} \bA \bA^T \label{eq:W_t}
\end{equation}

\noindent
In \eqref{eq:W_t}, the scalar $v_{q_t}$ models the variance of the error

\noindent
\begin{equation}
    \bq_t = \bs_t - \bx
\end{equation}

\noindent
Similarly, we define $v_{h_t}$ to be the variance of the error of the linear step

\noindent
\begin{equation}
    \bh_t = \br_t - \bx
\end{equation}

\noindent
Lastly, the scalar $\gamma_t$ follows the oracle identity

\noindent
\begin{equation}
    \gamma_t = -\frac{\bq_t^T \bA^T \bF_t(\bz_t)}{N v_{q_t}} \label{eq:gamma_general}
\end{equation}

\noindent
and can be practically estimated by rewriting $- \bA \bq_t = \bz_t - \bw$ and leveraging Lemma \ref{lemma:LSL_gamma_recursion}.

Next, we study certain properties of CG-VAMP under the following assumptions

\textbf{Assumption 1}: The dimensions of the signal model $N$ and $M$ approach infinity with a fixed ratio $\delta = \frac{M}{N} = O(1)$

\textbf{Assumption 2}: The measurement matrix $\bA$ is right-orthogonally invariant, such that in the SVD of $\bA = \bU \bS \bV^T$, the matrix $\bV$ is independent of other random terms and is uniformly distributed on the set of orthogonal matrices. The matrix $\bU$ is allowed to be any orthogonal matrix and the matrix $\bS^T \bS$ is allowed to have any Limiting Eigenvalue Distribution with compact support.

\textbf{Assumption 3}: The denoiser $\bg_D$ is uniformly Lipschitz so that the sequence of functions $\bg_D: \mathbb{R}^N \mapsto \mathbb{R}^N$ indexed by $N$ are Lipschitz continuous with a Lipschitz constant $L_N < \infty$ as $N \rightarrow \infty$ \cite{NS-VAMP}, \cite{AMP_SE_non_separable}. Additionally, we assume the sequences of the following inner-products are almost surely finite as $N \rightarrow \infty$ \cite{NS-VAMP}

\noindent
\begin{gather*}
    \lim_{N \rightarrow \infty} \inv{N} \bg_D(\bx + \bd_1)^T \bg_D(\bx + \bd_2), \; \; \lim_{N \rightarrow \infty} \inv{N} \bx^T \bg_D(\bx + \bd_1), \\ 
    \lim_{N \rightarrow \infty} \inv{N} \bd_1^T \bg_D(\bx + \bd_2), \quad \lim_{N \rightarrow \infty} \inv{N} \bx^T \bz_1, \quad  \lim_{N \rightarrow \infty} \inv{N} ||\bx||^2 
\end{gather*}

\noindent
where $\bd_1, \bd_2 \in \mathbb{R}^N$ with $( \bd_{1,n}, \bd_{2,n} ) \sim \normDensity(\bzero, C)$ for some positive definite $C \in \mathbb{R}^2$.

\noindent
Additionally, without the loss of generality we let $\bA$ be normalized so that $\inv{N} Tr \big\{ \bA \bA^T \big\} = 1$. Then, under the above assumptions, we have that \cite{UnifiedSE}, \cite{VAMP}

\noindent
\begin{gather}
    \bh_t \sim \normDensity(0, v_{h_t} \bI_N) \label{eq:h_t_gaussian} \\
    \bb_t = \bV^T \bq_t \sim \normDensity(0, v_{q_t} \bI_N) \label{eq:b_t_gaussian}
\end{gather}

\noindent
where 

\noindent
\begin{gather}
    \lim_{N \rightarrow \infty} \inv{N} \bh_{\tau}^T \bq_{k} \overset{a.s.}{=} 0  \label{eq:h_q_independence} \\
    \lim_{N \rightarrow \infty} \inv{N} \bw^T \bD \bb_{\tau} \overset{a.s.}{=} 0 \label{eq:w_Aq_independence}
\end{gather}

\noindent
for any $\tau, k \leq t$ and with $\bq_0 = - \bx$ and for any matrix $\bD \in \mathbb{R}^{M \times N}$ whose limiting spectral distribution has finite support. 

\section{Formulation of the problem}

In the above context, we aim to prove Theorem 3 from the main paper that studies the function $\partial(\halfa)$

\noindent
\begin{equation}
    \partial(\halfa) =  u_2 + 2 u_3 \halfa \label{eq:quadratic_equation_derivative}
\end{equation}

\noindent
with

\noindent
\begin{equation}
    u_2 = \frac{2}{N} \Big( \big( \br_t - \bg_D(\br_t) \big)^T \br_t - \big( \by - \bA \bg_D(\br_t) \big)^T \bA \br_t \Big) \quad \quad \quad u_3 = \inv{N} ||\br_t||^2 -\inv{N} ||\bA \br_t||^2
\end{equation}

\noindent
at the point $\halfa = \alpha_t$. The special case of that theorem for CG-VAMP is formulated as follows.

\begin{theorem} \label{th:divergence_identity_OAMP_based_algorithms}
    
    Consider the CG-VAMP algorithms equipped with a denoiser $\bg_D(\br_t)$ and let $\alpha_t$ be the divergence of $\bg_D(\br_t)$. Define a corrected denoiser
    
    \noindent
    \begin{equation}
        \bbf(\br_t) = \bg_D(\br_t) - \alpha_t \br_t \label{eq:th_DF_func}
    \end{equation}
    
    \noindent
    and its error
    
    \noindent
    \begin{equation*}
        \bbq_{t+1} = \bbf(\br_t) - \bx
    \end{equation*}
    
    \noindent
    Additionally define an inner-product $\psi_t = \inv{N} \bq_t^T \bbq_{t+1}$, where $\bq_t$ is the error $\bq_t = \bs_t - \bx$ from the previous iteration. Then, under Assumptions 1-3 and assuming \eqref{eq:h_t_gaussian} - \eqref{eq:w_Aq_independence} hold up to iteration $t$, the function $\frac{1}{2}\partial(\halfa)$ at $\halfa = \alpha_t$ almost surely converges to

    \noindent
    \begin{equation}
        \lim_{N \rightarrow \infty} \frac{1}{2} \partial(\alpha_t) \overset{a.s.}{=} \frac{k_t v_{q_t} - \psi_t}{\gamma_t v_{q_t}} - (v_{q_t} - \psi_t) - v_w + \frac{ \frac{\psi_t}{v_{q_t}} \Big( \lim_{N\rightarrow \infty} v_w \inv{N} \bw^T \bmu_t^i - \delta v_w + 2 v_w \gamma_t v_{q_t} \Big) }{ \gamma_t v_{q_t}} \label{eq:LSL_deriv_CG_VAMP}
    \end{equation}
    
    \noindent
    where $\bmu_t^i$ is the CG approximation with $0 \leq i \leq M$ iterations of the system of linear equations \eqref{eq:SLE} and $k_t = \frac{\inv{N}||\bA^T \bmu_t^i||^2}{\gamma_t}$.

\end{theorem}

Besides proving Theorem \ref{th:divergence_identity_OAMP_based_algorithms}, we provide the motivation for assuming $\partial(\halfa)$ is positive for CG-VAMP at any $t$, which implies validity of Proposal 1 from the main paper for CG-VAMP.

\section{Proof of Theorem \ref{th:divergence_identity_OAMP_based_algorithms}}

\begin{proof}[\unskip\nopunct]

Before proceeding to the proof, we note that while there are multiple version of the CG algorithm, in this document we assume CG-VAMP utilizes the following version of CG approximating the SLE \eqref{eq:SLE}

\begin{algorithm}
\DontPrintSemicolon
\SetNoFillComment
\caption{Conjugate Gradient for approximating $\bW \bmu = \bz$}

    \KwInitialization{ $\bmu = \bzero, \br^0 = \bp^0 = \bz, i = 0$}

    \While{$i < i_{max}$}
    {
        $\bd^i = \bW \bp^i$ \;
        $a^i = \frac{||\br^i ||^2}{ (\bp^i)^T \bd^i}$ \;
        $\bmu^{i+1} = \mu^i + a^i \bp^i$ \;
        $\br^{i+1} = \br^i - a^i \bd^i$ \;
        $b^i = \frac{|| \br^{i+1} ||^2}{|| \br^{i} ||^2}$ \;
        $\bp^{i+1} = \br^{i+1} + b^i \bp^i = \bz - \bW \bmu^{i+1} + b^i \bp^i$ \;
        $i = i+1$
    }
    \KwOutput{$\bmu^{i+1}$}    

\end{algorithm}

\noindent
Additionally, in the following, we will make use of the following lemmas

\begin{lemma} \label{lemma:normal_vectors_matrix_interaction}
    
    \cite{rand_mat_silverstein}, \cite{NS-VAMP}, \cite{AMP_SE_non_separable}: Let $\bx_1$ and $\bx_2$ be two $N$-sized zero-mean isotropic Gaussian vectors with variances $v_1$ and $v_2$ respectively and $\bD$ be a $N$ by $N$ symmetric positive semidefinite matrix independent of $\bx_1$ and $\bx_2$. Additionally assume that as $N \rightarrow \infty$, the Empirical Eigenvalue Distribution of $\bD$ converges to a density with compact support. Then almost surely we have that
    
    \noindent
    \begin{equation}
        \lim_{N \rightarrow \infty} \inv{N} \bx_1^T \bD \bx_2 \overset{a.s.}{=} \frac{\bx_1^T \bx_2}{N} \inv{N} Tr \big\{ \bD \big\} \label{eq:normal_vect_matrix_inner_prod}
    \end{equation}
    
\end{lemma}

\begin{lemma} \label{lemma:LSL_CG_model}
    
\cite{CG_EP}: Let $\bmu_t^i = \bF_t^i(\bz_t)$ be the output of the CG algorithm as in Algorithm 1 after $i$ inner-loop iterations approximating the SLE \eqref{eq:SLE}. Then, under Assumptions 1-3 and assuming \eqref{eq:h_t_gaussian} - \eqref{eq:w_Aq_independence} hold up to iterations $t$, the CG algorithm almost surely converges to a linear mapping

\noindent
\begin{equation}
    \lim_{N \rightarrow \infty} \bF_t^{i} \overset{a.s.}{=} \bU \btF_t^{i} \bU^T \label{eq:LSL_CG_model}
\end{equation}

\noindent
of the vector $\bz_t = \bw - \bA \bq_t$. In \eqref{eq:LSL_CG_model} the matrix $\btF_t^{i}$ is a non-singular diagonal matrix whose elements depend on $\delta$, $v_{q_t}$, $v_w$ and up to $i$ order moments of the singular spectrum of $\bS \bS^T$, and is independent of $\bq_t$ and $\bw$.

\end{lemma}

With these results at hand, we are ready to prove Theorem \ref{th:divergence_identity_OAMP_based_algorithms}. First, as was shown in the main paper, the function $\frac{1}{2}\partial(\halfa)$ at $\halfa = \alpha_t$ almost surely converges to

\noindent
\begin{equation}
    \lim_{N\rightarrow \infty} \frac{1}{2}\partial(\alpha_t) = v_{h_t} - v_w + \psi_t + \frac{\lim_{N\rightarrow \infty} \inv{N} (\bF_t(\bz_t))^T\bW_t  \bA \bbq_{t+1} + v_w \gamma_t \psi_t }{ \gamma_t v_{q_t}} \label{eq:starting_point}
\end{equation}

\noindent
with $\psi_t$ and $\bbq_{t+1}$ are defined as in Theorem \ref{th:divergence_identity_OAMP_based_algorithms}. Note that Lemma \ref{lemma:LSL_CG_model} implies that we can treat the output of the CG algorithm $\bmu_t^i = \bF_t^i(\bz_t)$ as the vector $\bz_t$ mapped by a certain matrix $\bF_t^i$. Then, based on Lemma \ref{lemma:normal_vectors_matrix_interaction}, we can rewrite the inner-product in \eqref{eq:starting_point} as

\noindent
\begin{equation*}
    \lim_{N \rightarrow \infty} \inv{N} \bbq_{t+1}^T \bA^T \bW_t \bmu_t^i \overset{a.s.}{=} - \psi_t \inv{N} Tr \Big\{ \bA \bA^T \bW_t \bF_t^i \Big\}
\end{equation*}

\noindent
This implies that

\noindent
\begin{equation}
    \lim_{N \rightarrow \infty} \inv{N} \bbq_{t+1}^T \bA^T \bW_t \bmu_t^i \overset{a.s.}{=} \lim_{N \rightarrow \infty} - \frac{\psi_t}{v_{q_t}} \inv{N} \bq_t^T \bq_t \inv{N} Tr \Big\{ \bA \bA^T \bW_t \bF_t^i \Big\} \overset{a.s.}{=} \lim_{N \rightarrow \infty} \frac{\psi_t}{v_{q_t}} \inv{N} \bq_t^T \bA^T \bW_t \bmu_t^i \label{eq:CG_VAMP_q_hat_A_W_H_z}
\end{equation}

\noindent
Next, by noting that $\bz_t = \by - \bA \bs_t = \bw - \bA \bq_t$, we can arrive at

\noindent
\begin{equation}
    - \bq_t^T \bA^T \bW_t \bmu_t^i = \bz_t \bW_t \bmu_t^i - \bw^T \bW_t \bmu_t^i \label{eq:CG_VAMP_q_A_W_H_z}
\end{equation}

\noindent
Since in the CG algorithm we have $\bp_t^0 = \bz_t$ and due to the conjugacy properties of CG, we have that

\noindent
\begin{equation}
     \bz_t \bW_t \bmu_t^i = \big( \bp_t^0 \big)^T \bW_t \sum_{j = 0}^i a_t^j \bp_t^j = a_t^0 \big( \bp_t^0 \big)^T \bW_t \bp_t^0 \overset{(a)}{=} ||\bz_t||^2 \label{eq:z_W_mu}
\end{equation}

\noindent
where (a) comes from the fact that $a_t^0 = \frac{||\bz_t||^2}{\big( \bp_t^0 \big)^T \bW_t \bp_t^0}$. Note that \eqref{eq:W_t} suggests that

\noindent
\begin{equation}
    \inv{N} \bw^T \bW_t \bmu_t^i = v_w \inv{N} \bw^T \bmu_t^i + v_{q_t} \inv{N} \bw^T \bA \bA^T \bmu_t^i
\end{equation}

\noindent
In \cite{OurSecondPaper} we showed that $\lim_{N \rightarrow \infty} \inv{N} \bw^T \bA \bA^T \bmu_t^i \overset{a.s.}{=} v_w \gamma_t$ so that

\noindent
\begin{equation}
    \lim_{N \rightarrow \infty} \inv{N} \bw^T \bW_t \bmu_t^i \overset{a.s.}{=} v_w \inv{N} \bw^T \bmu_t^i + v_w v_{q_t} \gamma_t \label{eq:w_W_mu}
\end{equation}

\noindent
Using \eqref{eq:CG_VAMP_q_A_W_H_z} in \eqref{eq:CG_VAMP_q_hat_A_W_H_z} together with \eqref{eq:z_W_mu}, \eqref{eq:w_W_mu} and the fact that $\lim_{N \rightarrow \infty} \inv{N} ||\bz_t||^2 \overset{a.s.}{=} v_{q_t} + \delta v_w$ \cite{OAMP}, \cite{EP_Keigo} leads us to

\noindent
\begin{align}
    \lim_{N \rightarrow \infty} \inv{N} \bbq_{t+1}^T \bA^T \bW_t \bmu_t^i &\overset{a.s.}{=} \lim_{N \rightarrow \infty} \frac{\psi_t}{v_{q_t}} \Big( v_w \inv{N} \bw^T \bmu_t^i + v_w v_{q_t} \gamma_t - v_{q_t} - \delta v_w \Big) \nonumber\\
    &= \lim_{N \rightarrow \infty} \frac{\psi_t}{v_{q_t}} \Big( v_w \inv{N} \bw^T \bmu_t^i - \delta v_w \Big) + v_w \gamma_t \psi_t - \psi_t
\end{align}

\noindent
Substituting this result into \eqref{eq:starting_point} gives

\noindent
\begin{equation}
    \lim_{N\rightarrow \infty} \frac{1}{2}\partial(\alpha_t) = v_{h_t} + \psi_t - \frac{\psi_t}{\gamma_t v_{q_t}} - v_w + \frac{\lim_{N\rightarrow \infty} \frac{\psi_t}{v_{q_t}} \Big( v_w \inv{N} \bw^T \bmu_t^i - \delta v_w \Big) + 2 v_w \gamma_t \psi_t }{ \gamma_t v_{q_t}} \label{eq:second_point}
\end{equation}

\noindent
Next, we can use \eqref{eq:gamma_general} to show that the variance $v_{h_t} \overset{a.s.}{=} \lim_{N \rightarrow \infty} \inv{N} ||\br_t - \bx||^2$ has the following update rule

\noindent
\begin{equation}
    v_{h_t} \overset{a.s.}{=} \lim_{N\rightarrow \infty} \gamma^{-2} \inv{N} ||\bA^T \bmu_t||^2 - v_{q_t} = \lim_{N\rightarrow \infty} \gamma_t^{-1} k_t - v_{q_t}
\end{equation}

\noindent
where we defined

\noindent
\begin{equation}
    k_t = \frac{\inv{N} ||\bA^T \bmu_t||^2}{\gamma_t} \label{eq:k_t_original}
\end{equation}

\noindent
Then, we can show that

\noindent
\begin{equation}
    v_{h_t} + \psi_t - \frac{\psi_t}{\gamma_t v_{q_t}} = \gamma_t^{-1} k_t- \frac{\psi_t}{\gamma_t v_{q_t}} - (v_{q_t} - \psi_t) = \frac{k_t v_{q_t} - \psi_t}{\gamma_t v_{q_t}} - (v_{q_t} - \psi_t)
\end{equation}

\noindent
Substituting this back into \eqref{eq:second_point} gives the desired result

\noindent
\begin{equation}
    \lim_{N\rightarrow \infty} \frac{1}{2}\partial(\alpha_t) = \frac{k_t v_{q_t} - \psi_t}{\gamma_t v_{q_t}} - (v_{q_t} - \psi_t) - v_w + \frac{ \frac{\psi_t}{v_{q_t}} \Big( \lim_{N\rightarrow \infty} v_w \inv{N} \bw^T \bmu_t^i - \delta v_w + 2 v_w \gamma_t v_{q_t} \Big) }{ \gamma_t v_{q_t}} \label{eq:final_point}
\end{equation}

\noindent
which completes the proof of Theorem \ref{th:divergence_identity_OAMP_based_algorithms} and of Theorem 3 for CG-VAMP from the main paper. \end{proof}

\section{Motivation for positivity of \eqref{eq:final_point}}

In this section we provide a motivation for assuming \eqref{eq:final_point} is positive, which is required for validity of Proposal 1 from the main paper for CG-VAMP. Note that when the SLE \eqref{eq:SLE} is very poorly or accurately approximated, the positivity of \eqref{eq:final_point} is ensured, as discussed in the main paper. Below, we consider the case, where the SLE \eqref{eq:SLE} is approximated with only a moderate accuracy. Importantly, in \cite{EP_Keigo} it was shown that when the SLE \eqref{eq:SLE} is solved exactly, \eqref{eq:gamma_general} is the largest and can be defined through the Limiting Eigenvalue distribution of $\bA \bA^T$ as

\noindent
\begin{equation}
    \lim_{N \rightarrow \infty} \gamma_t = \delta \int \frac{\lambda}{v_w + v_{q_t} \lambda} p(\lambda) d \lambda \label{eq:VAMP_gamma}
\end{equation}

\noindent
which implies that

\begin{equation}
    \lim_{N \rightarrow \infty} \gamma_t v_{q_t} = \delta \int \frac{v_{q_t} \lambda}{v_w + v_{q_t} \lambda} p(\lambda) d \lambda < \delta \label{eq:gamma_v_q_bound}
\end{equation}

\noindent
However, since we consider the case where the CG approximation of the SLE \eqref{eq:SLE} is only moderately accurate, we have the values of $\gamma_t$ to be even smaller then \eqref{eq:VAMP_gamma} and expect $\gamma_t v_{q_t}$ to be bounded by a substantially smaller value than $\delta$. 

On the other hand, by rewriting $-\bA \bq_t = \bz_t - \bw$, we can show that the scalar $\gamma_t$ from \eqref{eq:gamma_general} is equivalent to

\noindent
\begin{equation}
    \gamma_t = \lim_{N \rightarrow \infty}  \frac{ \bz_t^T \bmu_t^i - \bw^T \bmu_t^i}{N v_{q_t}}
\end{equation}

\noindent
Additionally, by using \eqref{eq:W_t} we can show that

\noindent
\begin{equation}
    ||\bA^T \bmu_t^i||^2 = \frac{(\bmu_t^i)^T\bW_t \bmu_t^i - v_w ||\bmu_t^i||^2}{v_{q_t}} \label{eq:norm_At_H_z}
\end{equation}

\noindent
The two above results imply that the scalar $k_t$ from \eqref{eq:k_t_original} almost surely converges to

\noindent
\begin{equation}
    \lim_{N \rightarrow \infty} k_t \overset{a.s.}{=} \lim_{N \rightarrow \infty} \frac{(\bmu_t^i)^T\bW_t \bmu_t^i - v_w ||\bmu_t^i||^2}{(\bmu_t^i)^T \bW_t \bmu_t^i - \bw^T \bmu_t^i } \label{eq:k_t_identity}
\end{equation}

\noindent
Next, we analyze and compare the terms involved in \eqref{eq:k_t_identity}. First, note that the term $(\bmu_t^i)^T \bW_t \bmu_t^i$ can be bounded as

\noindent
\begin{align}
    (\bmu_t^i)^T \bW_t \bmu_t^i &\geq \kappa_{min}(\bW_t) ||\bmu_t^i||^2 \label{eq:mu_W_mu_bound_1} \\
    &= \kappa_{min}\big( v_w \bI + v_{q_t} \bA \bA^T \big) ||\bmu_t^i||^2 = \big( v_w \bI + v_{q_t} \kappa_{min}(\bA \bA^T) \big) ||\bmu_t^i||^2 > v_w ||\bmu_t^i||^2 \label{eq:mu_W_mu_bound_2}
\end{align}

\noindent
Here, the equality $(\bmu_t^i)^T \bW_t \bmu_t^i = \kappa_{min}(\bW_t) ||\bmu_t^i||^2$ corresponds to the case, where the vector $\bmu_t^i$ aligns with the smallest eigenvalue of the matrix $\bW_t$. But in the limit, the chance of that happening is infinitely small due to the nature of the vector $\bmu_t^i$. From the literature \cite{lin_nonlin_program} on the CG algorithm, we know that the output $\bmu_t^i$ of the CG algorithm with $i$ iterations can be represented as a polynomial $P_i(\bW_t)$ of the matrix $\bW_t$ times the input vector $\bz_t$

\noindent
\begin{equation}
    \bmu_t^i = \sum_{j = 0}^i \phi_j \bW_t^j \bz_t = P_i(\bW_t) \bz_t
\end{equation}

\noindent
for some scalars $\phi_j$. Note that in the limit, $\bz_t$ acts as $\bw + \bU \bS \bb_t$, where $\bw$ and $\bb_t$ are zero-mean i.i.d. Gaussian vectors independent of each other and of $\bW_t$. Then, by defining a matrix $\bD_t^i = P_i(\bW_t) \bW_t P_i(\bW_t)$ and using Lemma \ref{lemma:normal_vectors_matrix_interaction}, we have that the inner-product $(\bmu_t^i)^T \bW_t \bmu_t^i = \bz_t^T \bD_t^i \bz_t$ almost surely converges to

\noindent
\begin{equation}
    \lim_{N \rightarrow \infty} \inv{N} (\bmu_t^i)^T \bW_t \bmu_t^i = \lim_{N \rightarrow \infty} \inv{N} \bz_t^T \bD_t^i \bz_t \overset{a.s.}{=} \lim_{N \rightarrow \infty} v_w \inv{N} Tr \big\{ \bD_t^i \big\} + v_{q_t} \inv{N} Tr \big\{ \bA \bA^T \bD_t^i \big\} = \inv{N} Tr \big\{ \bW_t \bD_t^i \big\} \label{eq:mu_W_mu_LSL}
\end{equation}

\noindent
We see that in the limit, the vector $\bmu_t^i$ interacts with the eigenvalues of the matrix $\bW_t$ in the averaged way through the trace operator, while \eqref{eq:mu_W_mu_bound_1} assumes that $\bmu_t^i$ aligns with the smallest eigenvalue. Then, if the matrix $\bA$ is sufficiently ill-conditioned, it is reasonable to assume that \eqref{eq:mu_W_mu_LSL} is much greater than \eqref{eq:mu_W_mu_bound_1}. Otherwise, doing a few iterations of CG would be sufficient to get an accurate approximation of the SLE \eqref{eq:SLE} and therefore we should rather refer to the asymptotic result for VAMP considered in the main paper. Therefore, in the following we assume $v_w ||\bmu_t||^2$ has a negligible magnitude compared to $(\bmu_t^i)^T \bW_t \bmu_t^i$ and therefore

\noindent
\begin{equation}
    \lim_{N \rightarrow \infty} k_t \overset{a.s.}{\approx} \lim_{N \rightarrow \infty} \frac{(\bmu_t^i)^T\bW_t \bmu_t^i }{(\bmu_t^i)^T \bW_t \bmu_t^i - \bw^T \bmu_t^i } \label{eq:k_t_approximation}
\end{equation}

To proceed next, we introduce the following lemma

\begin{lemma} \label{lemma:w_mu_positivity}
    Consider the CG-VAMP algorithm. Define $\bmu_t^i$ to be the CG approximation with $0 \leq i \leq M$ iterations of the system of linear equations \eqref{eq:SLE}. Then, under Assumptions 1-3 and assuming \eqref{eq:h_t_gaussian} - \eqref{eq:w_Aq_independence} hold up to iterations $t$, we almost surely have that
    
    \noindent
    \begin{equation}
        \lim_{N \rightarrow \infty} \inv{N} \bw^T \bmu_t^{i} \overset{a.s.}{>} 0 \label{eq:w_mu_positivity}
    \end{equation}
\end{lemma}

\begin{proof}
    See Appendix A
\end{proof}

\noindent
This lemma together with \eqref{eq:k_t_approximation} suggests that $k_t$ is either tightly concentrated around $1$ if $\bw^T \bmu_t^i$ has a negligible magnitude comparing to $(\bmu_t^i)^T \bW_t \bmu_t^i$ or is evidently greater than $1$ if the two terms are of the same order. For the first case we can use \eqref{eq:gamma_v_q_bound} to show that \eqref{eq:starting_point} is tightly concentrated around

\noindent
\begin{equation}
    \lim_{N\rightarrow \infty} \frac{1}{2}\partial(\alpha_t) \overset{a.s.}{\approx} \frac{(v_{q_t} - \psi_t) - \gamma_t v_{q_t} (v_{q_t} - \psi_t) - \gamma_t v_{q_t} v_w}{\gamma_t v_{q_t}} + \frac{ \frac{\psi_t}{v_{q_t}} \Big( \lim_{N\rightarrow \infty} v_w \inv{N} \bw^T \bmu_t^i - \delta v_w + 2 v_w \gamma_t v_{q_t} \Big) }{ \gamma_t v_{q_t}} \label{eq:partial_approx}
\end{equation}

\noindent
Since $\gamma_t v_{q_t}$ is positive and we are interested in showing that \eqref{eq:starting_point} is positive, we can consider only the nominator and bound it by dropping the component $\lim_{N\rightarrow \infty} v_w \inv{N} \bw^T \bmu_t^i$, which is positive due to Lemma \ref{lemma:w_mu_positivity}. Thus we shall analyze the following difference

\noindent
\begin{align}
    (v_{q_t} - \psi_t) (1 - \gamma_t v_{q_t}) - v_w \big( \delta \frac{\psi_t}{v_{q_t}} + \gamma_t v_{q_t} - 2 \gamma_t \psi_t \big) &= (v_{q_t} - \psi_t) (1 - \gamma_t v_{q_t}) - v_w \big( \delta \frac{\psi_t}{v_{q_t}} + 2\gamma_t (v_{q_t} - \psi_t) - \gamma_t v_{q_t} \big) \nonumber\\
    &= (v_{q_t} - \psi_t) (1 - \gamma_t v_{q_t} - 2 \gamma_t v_w) - v_w \big( \delta \frac{\psi_t}{v_{q_t}} - \gamma_t v_{q_t} \big)
\end{align}

\noindent
Let $v_{q_t} > v_w$ so that

\noindent
\begin{equation}
    1 - \gamma_t v_{q_t} - 2 \gamma_t v_w > 1 - \gamma_t v_{q_t} - 2 \gamma_t v_{q_t} = 1 - 3 \gamma_t v_{q_t}
\end{equation}

\noindent
As discussed at the beginning of the section, we expect $\gamma_t v_{q_t}$ to be considerably smaller than $\delta$, which implies that $ 1 - 3 \gamma_t v_{q_t} $ is tightly concentrated around $1$. Similarly, we expect that

\noindent
\begin{equation}
    v_w \big( \delta \frac{\psi_t}{v_{q_t}} - \gamma_t v_{q_t} \big) \approx v_w  \frac{\psi_t}{v_{q_t}} \delta
\end{equation}

\noindent
As a result, we the nominator of \eqref{eq:partial_approx} could be approximated by

\noindent
\begin{equation}
    v_{q_t} - \psi_t - v_w  \frac{\psi_t}{v_{q_t}} \delta \label{eq:last_step}
\end{equation}

\noindent
In the main paper we provided the motivation for assuming $v_{q_t}$ is greater than $(\delta^{-1} - 1)^{-1} v_w$ in the context of VAMP. Here, the scaling $\frac{\psi_t}{v_{q_t}} \delta$ of $v_w$ is even smaller, since $\psi_t < v_{q_t}$. Additionally, the fixed point of CG-VAMP is assumed to be worse than of that of VAMP so that $v_{q_t}$ is even larger. All together, this motivates assuming \eqref{eq:last_step} is positive, which implies positivity of \eqref{eq:starting_point}. Note that in the case where $k_t$ is greater than $1$, the analysis above holds automatically, since $v_{q_t}$ in \eqref{eq:last_step} is additionally scaled by $k_t>1$.

\section*{Appendix A} 
\begin{proof}[\unskip\nopunct]

Next we prove Lemma \ref{lemma:w_mu_positivity}. The proof is based on the following lemma from \cite{OurPaper}

\begin{lemma}
\label{lemma:LSL_gamma_recursion}

\cite{OurPaper} For $i = 0,1,...$ let the scalars $a_t^{i}$ and $b_t^{i}$ be as in the CG algorithm shown in Algorithm 1 and define recursively two scalar functions $\eta_t^{i}$ and $\psi_t^{i}$ as follows

\noindent
\begin{gather}
    \eta_t^{i} = v_w \Big( \delta - \inv{N} \bz_t^T \bmu_t^i \Big) + b_t^{i-1} \eta_t^{i-1} \label{eq:LSL_gamma_A_rec_1}\\
    \psi_t^i = \psi_t^{i-1} + a_t^{i-1} \eta_t^{i-1} \label{eq:LSL_gamma_A_rec_2}
\end{gather}

\noindent
with $\psi_t^0 = 0$ and $\eta_t^0 = \delta v_w$. Then under Assumptions 1-3 and assuming \eqref{eq:h_t_gaussian} - \eqref{eq:w_Aq_independence} hold up to iterations $t$, the inner product $\inv{N} \bw^T \bmu_A^{i}$ almost surely converges to

\noindent
\begin{equation*}
    \lim_{N \rightarrow \infty} \inv{N} \bw^T \bmu_A^{i} \overset{a.s.}{=} \psi_t^i
\end{equation*}

\end{lemma}

\noindent
From Algorithm 1 we see that both $a_t^i$ and $b_t^i$ are positive. Then, by induction, we can prove that $\psi_t^i$ is non-negative if we can prove that $\delta - \inv{N} \bz_t^T \bmu_t^i$ is non-negative. First, note that

\noindent
\begin{equation}
     \bz_t^T \bmu_t^i = (\bmu_t^{M})^T \bW_t \bmu_t^i = \big( \sum_{j = 0}^M a_t^j \bp_t^j \big)^T \bW_t \sum_{j = 0}^i a_t^j \bp_t^j = \big( \sum_{j = 0}^i a_t^j \bp_t^j \big)^T \bW_t \sum_{j = 0}^i a_t^j \bp_t^j = (\bmu_t^{i})^T \bW_t \bmu_t^i
\end{equation}

\noindent
Since this inner-product monotonically increases with $i$ \cite{CG_properties}, it is sufficient to consider the case $i=M$, where $\bmu_t^M = \bW_t^{-1} \bz_t$

\noindent
\begin{align}
    &\lim_{N \rightarrow \infty} \inv{N} \bz_t^T \bW_t^{-1} \bz_t \nonumber\\
    &\overset{a.s.}{=} \lim_{N \rightarrow \infty} v_w \inv{N} Tr \big\{ \bW_t^{-1} \big\} +  v_{q_t} \inv{N} Tr \big\{ \bA \bA^T \bW_t^{-1} \big\} \label{eq:z_invW_z_intermidiate}\\
    &= \inv{N} Tr \big\{ (v_w + v_{q_t} \bA \bA^T ) \bW_t^{-1} \big\} = \delta \nonumber
\end{align}

\noindent
where \eqref{eq:z_invW_z_intermidiate} comes from Lemma \ref{lemma:normal_vectors_matrix_interaction} and \eqref{eq:w_Aq_independence}. Thus, we have that

\noindent
\begin{equation}
    \inv{N} \bz_t^T \bmu_t^i \leq \delta
\end{equation}

\noindent
with equality when $i=M$. This, together with Lemma \ref{lemma:LSL_gamma_recursion}, completes the proof. \end{proof}

\printbibliography